\documentclass{amsart}
\usepackage{hyperref}
\usepackage{graphicx}
\usepackage{curves}
\usepackage{tikz}
\usetikzlibrary{arrows}
\usepackage{enumerate}
\usepackage{float}
\newcommand*{\mailto}[1]{\href{mailto:#1}{\nolinkurl{#1}}}

\newtheorem{theorem}{Theorem}[section]
\newtheorem{lemma}[theorem]{Lemma}
\newtheorem{corollary}[theorem]{Corollary}
\newtheorem{remark}[theorem]{Remark}
\newtheorem{hypothesis}[theorem]{Hypothesis}

\newcommand{\R}{\mathbb{R}}
\newcommand{\Z}{\mathbb{Z}}

\newcommand{\C}{\mathbb{C}}

\newcommand{\nn}{\nonumber}
\newcommand{\beq}{\begin{equation}}
\newcommand{\eeq}{\end{equation}}
\newcommand{\bea}{\begin{eqnarray}}
\newcommand{\eea}{\end{eqnarray}}

\newcommand{\ol}{\overline}
\newcommand{\pa}{\partial}
\newcommand{\ti}{\tilde}

\newcommand{\id}{\mathbb{I}}
\newcommand{\I}{\mathrm{i}}
\newcommand{\E}{\mathrm{e}}

\newcommand{\re}{\mathop{\mathrm{Re}}}
\newcommand{\im}{\mathop{\mathrm{Im}}}

\newcommand{\Sini}{\Sigma^{\text{ini}}}
\newcommand{\mini}{m^{\text{ini}}(k)}

\DeclareMathOperator{\res}{Res}

\newcommand{\si}{\sigma}
\newcommand{\la}{\lambda}
\newcommand{\ga}{\gamma}
\newcommand{\om}{\omega}

\numberwithin{equation}{section}

\newcommand{\sigI}{\begin{pmatrix} 0 & 1 \\ 1 & 0 \end{pmatrix}}
\newcommand{\rI}{\begin{pmatrix}  1 & 1 \end{pmatrix}}

 \newcommand{\noprint}[1]{}

\begin{document}

\title[KdV shock waves]{Asymptotics of KdV shock waves via the Riemann--Hilbert approach}

\author[I. Egorova]{Iryna Egorova}
\address{B. Verkin Institute for Low Temperature Physics \\ and Engineering of National Avcademy\\ of Sciences of Ukraine\\ 47,Lenin ave\\ 61103 Kharkiv\\ Ukraine}
\email{\href{mailto:iraegorova@gmail.com}{iraegorova@gmail.com}}

\author[M. Piorkowski]{Mateusz Piorkowski}
\address{Faculty of Mathematics\\ University of Vienna\\
Oskar-Morgenstern-Platz 1\\ 1090 Wien}
\email{\href{mailto:Mateusz.Piorkowski@univie.ac.at}{Mateusz.Piorkowski@univie.ac.at}}

\author[G. Teschl]{Gerald Teschl}
\address{Faculty of Mathematics\\ University of Vienna\\
Oskar-Morgenstern-Platz 1\\ 1090 Wien}
\email{\href{mailto:Gerald.Teschl@univie.ac.at}{Gerald.Teschl@univie.ac.at}}
\urladdr{\href{http://www.mat.univie.ac.at/~gerald/}{http://www.mat.univie.ac.at/\string~gerald/}}

\thanks{Research supported by the Austrian Science Fund (FWF) under Grants No.\ P31651 and W1245.}

\begin{abstract}
This paper discusses some general aspects and techniques associated with the long-time asymptotics of steplike solutions of the Korteweg--de Vries (KdV) equation via vector Riemann--Hilbert problems. We also elaborate on an ill-posedness of  the matrix Riemann--Hilbert problem for the KdV case in the class of matrices with square integrable singularities.  Furthermore, we refine the asymptotics for the shock wave in the Whitham zone derived previously and rigorously justify it for a more general class of initial data. In particular, we clarify the influence of resonances and of the discrete spectrum on the leading asymptotics.
\end{abstract}

\maketitle

\section{Introduction}

The nonlinear steepest descent (NSD) analysis for oscillatory Riemann--Hilbert problems (RHP) is a versatile tool in asymptotic analysis. This procedure naturally starts from
a reformulation of the original scattering problem as a Riemann--Hilbert factorization problem. In most cases this will be a matrix RHP as these are typically more convenient to analyze. Indeed, the fact
that a nonsingular solution can be used to cancel jumps on certain parts of the contour is a crucial trick which lies at the heart of the theory. However, for some problems, most prominently the Korteweg--de Vries equation
\beq\label{kdv}
q_t(x,t)=6q(x,t)q_x(x,t)-q_{xxx}(x,t), \quad (x,t)\in\R\times\R_+,
\eeq
it turned out that a vector RHP is the right choice. This is related to the fact that even in the simplest case of a single soliton there is a nontrivial solution of the
associated vanishing problem (see \cite{GT}). However, this is in contradiction to the classical uniqueness result for matrix RHPs and shows that the matrix problem cannot
have a solution in this situation. The remedy, as pointed out in \cite{GT}, is to work with the vector RHP and impose an additional symmetry condition to retain uniqueness.

Next, recall that the asymptotic analysis of such a RHP usually consists of three steps: The first step deforms the problem in such a way that the leading asymptotic contribution
is revealed. In the second step the parts of the jump which are expected not to contribute to the leading asymptotics are dropped, yielding a model problem which then needs to
be solved explicitly. In most cases, it is possible to find a matrix solution to this model problem and hence the final step, namely showing that the solution of the model problem
indeed asymptotically approximates the solution of the original RHP, can be performed using the well-established tools for matrix problems. However, for model problems leading to explicit solutions in terms of Jacobi theta functions, finding a nonsingular\footnote{In what follows \textit{nonsingular} refers to the matrix solution being invertible with at most $L^2$-integrable singularities on either side of the jump contour} matrix solution is not always possible (see \cite[Thm.~5.6]{BKT16}, \cite[Sect.~3]{BC12}, \cite[Sect.~3]{GGM}).

The main purpose of the present note is to study in depth such a RHP coming from the KdV equation, having only singular matrix model solutions for certain exceptional values of the parameters $x$ and $t$. Indeed for these values, the initial and the model problems do not have invertible bounded matrix solutions with admissible square integrable singularities in the points of discontinuity of the contour. We will refer to this feature as the ill-posedness of the matrix RHP for the KdV equation. 

The specific example that we will consider is the RHP associated with the long-time asymptotical behavior of shock waves for the KdV equation. Here the KdV shock problem is the Cauchy problem for  \eqref{kdv} with initial data $q(x,0)=q(x)$ satisfying:
\beq \label{ini}
\left\{ \begin{array}{ll} q(x)\to 0,& \ \ \mbox{as}\ \ x\to +\infty,\\
q(x)\to - c^2,&\ \ \mbox{as}\ \ x\to -\infty, \quad c>0.\end{array}\right.
\eeq
We recall that the asymptotic behavior of the shock wave  was first described on a physical level of rigor in the pioneering works of Gurevich and Pitayevskii \cite{GP}, \cite{gp2}. By applying  the Whitham approach  to the pure step initial data ($q(x)=0$ for $x>0$ and $q(x)=-c^2$  for $x\leq 0$), the authors derived the leading asymptotics in terms of a modulated elliptic wave.  For arbitrary steplike initial data \eqref{ini}  the analogous asymptotic term was calculated in \cite{EGKT} and \cite{EGT} using the NSD method.  In particular, it was shown that in the elliptic zone $ -6c^2 t<x<4c^2 t$ the shock wave is expected to be close to a modulated one gap solution of the KdV equation as $t\to\infty$. However, this has not been rigorously justified until now.
 
The main result of this paper is the completion of the asymptotic analysis for the shock wave in the Whitham zone, in the framework of the standard NSD method. Even though the inverse scattering transform for the KdV equation is given in terms of a vector RHP, the NSD approach involves building a matrix solution to the model RHP in order to match it with the local parametrix solutions. Since
the nonsingular  matrix model solution does not exist for certain arbitrary large pairs $x$ and $t$, we will instead use a singular matrix model solution which, despite its singular behavior, can be used to bound the error term in the asymptotics, as shown in \cite[Sect.~3]{GGM}. Note that for decaying initial data or rarefaction waves, meaning $q(x)\to 0$ as $x\to +\infty$ and $q(x)\to c^2$ as $x\to -\infty$, the nonsingular matrix model solutions always exist (see \cite{AELT}, \cite{GT}). 

As for the shock wave case, to characterize the pairs $(x,t)$ for which the nonsingular matrix model solution fails to exist, we must recall the trace formula for a finite gap KdV  solution. Denote by $\xi=\frac{x}{12t}$ the slowly varying parameter and consider values $(x,t)$ satisfying
\beq\label{domain}\xi\in \mathcal I_\varepsilon:=[-\frac{c^2}{2}+\varepsilon,\frac{c^2}{3}-\varepsilon],\eeq for an arbitrary small $\varepsilon>0$. Then,  as is shown in  \cite{GP}, \cite{EGKT}, there exists a smooth monotonously increasing positive function $a=a(\xi)$ such that $a(-\frac{c^2}{2})=0$ and $a(\frac{c^2}{3})=c $. This function characterizes the Whitham zone of the modulated  elliptic wave  $q^{mod} (x,t,\xi)$, which is 
 the periodic one gap solution of the KdV equation on the ray $\xi=\mbox{const}$. This one gap solution is associated with the spectrum \beq\label{spectr}\mathfrak G(\xi):=[-c^2, -a^2(\xi)] \cup \R_+,\eeq and with the initial Dirichlet divisor $(\la(0,0,\xi), \pm)$ defined via the scattering data of the potential \eqref{ini} by the formulas \eqref{invJ} and \eqref{deff}-\eqref{svyazF} below. 
Let $\la(x,t,\xi)\in [-a^2(\xi), 0\,]$ be the solution of the Dubrovin equations (\cite[Ch.~12]{Lev}) corresponding to the initial value $(\la(0,0,\xi), \pm)$. Then the well-known trace formula implies
 \[
 q^{mod}(x,t,\xi)=-c^2 - a(\xi)^2 -2\la(x,t,\xi).
 \]
We will show (see Remark \ref{rem33}) that the set of local minima of $q^{mod}(x,t,\xi)$:
\[
\mathcal O(\xi)=\{ (x,t):\ \la(x,t,\xi)=0\},
\]
coincides with the set of points where the associated matrix model problem has no nonsingular  solution. 
Evidently, these pairs $(x,t)$ appear for each $\xi\in I_\varepsilon$ and for arbitrary large $t$.

In turn, the circumstances which lead to the ill-posedness of the  initial matrix RHP associated  with the shock wave for certain (arbitrary large) points $(x,t)$ are the following.   Let $\phi(k,x,t)$ be the right Jost solution to the underlying spectral equation of the problem \eqref{kdv}--\eqref{ini}:
\beq\label{ell}L(t)y=-\frac{d^2}{dx^2} y + q(x,t) y=k^2 y,\eeq normalized as
\beq\label{limsi}
\lim_{x \to  +\infty} \E^{- \I kx} \phi(k,x,t) =1.
\eeq
In Section~\ref{sec:inivrhp} we show that if $\phi(0,x,t)=0$ for a  pair  $(x,t)$, then the nonsingular matrix solution for
the initial RHP does not exist. In connection with this observation an additional spectral problem appears: to find conditions which would guarantee that the right Jost solution associated with the shock wave  is nonzero at the edge of the continuous spectrum for sufficiently large $x$ and $t$ with
$(x,t)\in\mathcal D_\varepsilon$, where
\beq\label{DE}
\mathcal D_\varepsilon:=\{ (x,t)\in\R\times\R_+:\ \frac{x}{12 t}\in \mathcal I_\varepsilon\}. 
\eeq
It should be noted that the same condition $\phi(0,x,t)=0$ leads to the ill-posedness of the matrix RHP in the decaying case $q(x,t)\to 0$, $x\to\pm\infty$ too. 
Unlike in the steplike case, here we can propose sufficient conditions which guarantee  $\phi(0,x,t)\neq 0$. Indeed, assuming that the  discrete spectrum is absent in the decaying case, the Jost solutions are positive  below the spectrum (cf.\ \cite[Corollary 2.4]{GST}) and hence also at the boundary of the spectrum
$k=0$ by continuity (note that the zeros of a nontrivial solution of a Sturm--Liouville equation must always be simple). Thus a nonsingular matrix solution always exists
in this situation (this also follows from \cite[Theorem~9.3]{Zhou}). However, in the presence of discrete spectrum this is no longer true.

Our main result is the following

\begin{theorem} \label{theormain}
Let $q(x,t)$ be the unique solution of the initial value problem \eqref{kdv}--\eqref{ini} with the initial data satisfying
\beq\label{decay}
\int_0^{+\infty} \E^{\eta x}(|q(x)| + |q(-x)+c^2|)dx<\infty,\quad \  x^4 q^{(i)}(x)\in L^1(\R),\quad i=1,...,7,
\eeq
for some positive $\eta > 0$.
For any $\xi\in \mathcal I_\varepsilon$ with $\varepsilon > 0$ (see \eqref{domain}),  let $a=a(\xi)\in (0,c)$ be  defined implicitly by 
\beq\label{defa}
\int_0^{\I a}\left(k^2 + \xi +\frac{c^2-a^2}{2}\right)\sqrt{\frac{k^2 + a^2}{k^2 + c^2}}\  dk =0.
\eeq 
Let $p_0=p_0(\xi)$ be the point on the two-sheeted Riemann surface associated with  $\mathfrak G(\xi)$ (see \eqref{spectr}),
uniquely defined via the Jacobi inversion problem
\beq\label{defp0}\int_{-a^2}^{p_0} \frac{d\la}{\sqrt{\lambda (\lambda +c^2)(\lambda+a^2)}}=
\I \Delta (\xi),\eeq
with
\beq\label{deltaini}
\Delta(\xi)=\frac{\int_{\I a}^{\I c}\frac{2\log\left|T(s)\prod_{j=1}^N \frac{s-\I\kappa_j}{s+\I\kappa_j}\right| +\log\left|\frac{ s+\I c \ell }{s}\right|}{\left| (s^2 +c^2)(s^2+a^2)\right|^{1/2}}ds}{\int_{0}^{\I a}\left( (s^2 +c^2)(s^2+a^2)\right)^{-1/2}ds}-\frac {\pi\ell}{ 2},
\eeq 
where:
\begin{itemize}
\item $T(k)$ is the right transmission coefficient for the initial datum \eqref{decay};
\item $-\kappa_1^2<...<-\kappa_N^2$ is the discrete spectrum of the problem; 
\item
 $\ell=-1$ if the initial datum has a resonance at the point $\I c$, and $\ell=1$ in the general (nonresonant) case. 
\end{itemize}
Let $q^{mod}(x,t,\xi)$ be the periodic (one gap) solution to the KdV equation associated with the spectrum $\mathfrak G(\xi)$ and  the initial Dirichlet divisor $p_0=(\la(0,0,\xi), \pm)$. Then for all  $x\to\infty$, $t\to +\infty$ such that  $(x,t)\in\mathcal D_\varepsilon$, the following asymptotics is valid: 
\beq\label{as}
q(x,t)=q^{mod}(x,t,\frac{x}{12 t}) + O(t^{-1}),
\eeq
where the error term $O(t^{-1})$ is uniform with respect to $\xi\in \mathcal I_\varepsilon$.
\end{theorem}

Formula \eqref{as} is obtained in the framework of a standard NSD approach applied to a vector RHP. It includes some transformations (conjugations and deformations) which lead to an equivalent RHP with the jump matrix asymptotically close, as $t\to \infty$, to an exactly solvable model RHP except in small vicinities of two extreme points $\pm \I a(\xi)$. The approach also involves the construction of a proper matrix model solution  and an associated  matrix solution of the local parametrix problems. However, when performing this analysis in the KdV steplike  case, it is essential to take into account some specific features of the vector RHP. Note that unlike the matrix RHP, the proof of uniqueness for a vector RHP is typically more sophisticated and depends on particular properties of the jump matrix and of the contour, as well as on the class of admissible singularities for the solution. That is why it seems important for us   to  perform NSD  deformations and conjugations in a way that does not affect this uniqueness. To this end, in each transformation we impose additional symmetry assumptions on the contour, on the jump matrix and on the solution itself, including the model problem solution (see Hypothesis~\ref{remsym}). 

The solution $(m_1(k,x,t), m_2(k,x,t))$ of the initial RHP is unique (see Theorem \ref{thm:vecrhp}) and satisfies the aforementioned symmetry assumption. This symmetry requirement implies a symmetry of the "error vector", which in turn, allows  us to apply a new formula \beq\label{main1}
q(x,t)=\lim_{k\to\infty} 2 k^2 \left(m_1(k,x,t)m_2(k,x,t) -1\right)
\eeq for computing the leading term of the asymptotics, and this essentially simplifies the final asymptotical analysis. 

Note that the   traditional formula which connects the potential $q(x,t)$ with the solution of the initial RHP {\it (i)--(iii)}, Theorem \ref{thm:vecrhp} is the following one:
\beq\label{previous}
\frac{\partial}{\partial x}\lim_{k\to\infty} 2\I k (m_1(k,x,t) - 1)= q(x,t).
\eeq
Formula \eqref{main1} not only avoids the necessity to justify the differentiation with respect to $x$ in the asymptotical expansion for $m_1(k,x,t)$, but also allows us to extract the asymptotics
from the solution of the model vector RHP in a shorter and more transparent way  (see Section~\ref{sec4}) compared to \cite{EGKT}, \cite{EGT} and \cite{GGM}. In particular, this approach allows us to apply  the trace formulas when computing asymptotics.

\section{Well-posedness  of the initial (meromorphic) vector RHP}

In this section we  recall the statement of the initial vector RHP for the KdV shock wave (see \cite{EGKT}) and prove its well-posedness.  Note that in the present study we weaken the decay conditions on the initial data compared to \cite{EGKT}, where it is assumed that \[|q(x)| + |q(-x)+c^2|=O(\E^{-(c+\eta)x}), \quad x\to +\infty,\quad \eta>0.\]  

We choose the still quite restrictive condition \eqref{decay} to avoid complications with the analytical continuation of the scattering data in the framework of the NSD method. However, \eqref{decay} also guarantees the existence  of the unique classical solution $q(x,t)$ for the Cauchy problem \eqref{kdv}--\eqref{ini} (cf. \cite{EMT5, gladka}) satisfying 
\beq\label{first}
\int_0^{+\infty}|x|(|q(x,t)| + |q(-x,t)+c^2|)dx<\infty, \qquad t\in\R.
\eeq
In turn, this means that the use of the inverse scattering transform for the formulation of the respective RHP is well grounded. 

We start with recalling some well known facts of the scattering theory for the step-like Schr\"odinger operator \eqref{ell} with emphasis on analytical properties of the scattering data due to \eqref{decay} and with a detailed description of the influence of resonance on them. 

 The spectrum of the operator \eqref{ell} with potential \eqref{first}
 consists of an absolutely
continuous part $[-c^2,\infty)$ plus a finite number of eigenvalues  $-\kappa_j^2\in(-\infty,-c^2)$,
$1\le j \le N$ enumerated as in Theorem \ref{theormain}.

 Let $\phi(k,x,t)$ be the right Jost solution of \eqref{ell} satisfying \eqref{limsi} and let  $\phi_1(k,x,t)$
be the  Jost solution asymptotically close to the free exponent associated with the left background:\beq\label{lims}
 \lim_{x \to  -\infty} \E^{ \I k_1 x} \phi_1(k,x,t) =1, \qquad k_1:=\sqrt{k^2 +c^2}.
\eeq
Here $k_1>0$ for $k\in[0,\I c)_r$. The last notation denotes the right side of the cut along the interval $[0,\I c]$. Accordingly, $k_1<0$ for $k\in[0,\I c)_l$,  the left side of the cut.
The left Jost solution admits the usual representation via the transformation operator (\cite{MAR}):
\[
\phi_1(k,x,t)=\E^{-\I k_1 x} +  \int^x_{-\infty}K_1(x,y,t)\E^{-\I k_1 y}dy, 
\]   
where $K_1(x,y,t)$ is a real-valued function with 
\beq\label{1}
|K_1(x,y,t)|\leq C\int^{\frac{x+y}{2}}_{-\infty}
|q(s,t)+c^2| ds.
\eeq
Note that the function $\phi(k,x,t)$ is a holomorphic function of $k$ in $\C^+:=\{ k \in \C \colon \im k > 0\}$ and continuous up to the real axis. It  is real-valued for $k\in [0,\I c]$, and does not have a discontinuity on this interval. As for the function $\phi_1(k,x,t)$, it is holomorphic in the domain $\C^+\setminus (0, \I c]$ and continuous up to the boundary, where
$[\phi_1(k,x,t)]_r=[\overline{\phi_1(k,x,t)]}_l\,$ for $k\in[0,\I c]$. 

We observe  that condition \eqref{decay} together with \eqref{1} imply that for $t=0$ the second left Jost solution: 
\[\breve{\phi_1}(k,x,0)=\E^{\I k_1 x} +  \int^x_{-\infty}K_1(x,y,0)\E^{\I k_1 y}dy,\]
defined  for $k_1\in \R$, where $\breve{\phi_1}=\overline{ \phi_1}$,  admits  an analytical continuation into the domain  \[\mathcal V=\{ k:\, \re k_1\in [-c, c],\quad 0<\im k_1(k) <  \eta \}.\]
Note that $\mathcal V$ is a neighbourhood of the interval $ [\I c, 0)$.
Then the limiting values satisfy
\beq\label{soed} [\breve{\phi_1}(k,x,0)]_r=[\phi_1(k,x,0)]_l,\quad  [\breve{\phi_1}(k,x,0)]_l= [\phi_1(k,x,0)]_r,\ \ \mbox{for}\  k\in [0,\I c].\eeq
For $k\in\mathcal V$ introduce two Wronskians :
\[\aligned
W(k)&= \phi_1(k,x,0)\phi^\prime(k,x,0) -\phi_1^\prime(k,x,0)\phi(k,x,0);\\  \breve W(k)&= \breve\phi_1(k,x,0)\phi^\prime(k,x,0) -\breve \phi_1^\prime(k,x,0)\phi(k,x,0),\endaligned
\] where $f^\prime=\frac{\pa}{\pa x} f$. 
Then by \eqref{soed}
\beq\label{soed2} [W(k)]_r= [\breve W(k)]_l=[\overline{W(k)}]_l=[\overline{\breve W(k)}]_r.\eeq
The Wronskian $W(k)$ of the Jost solutions is in fact a holomorphic function in $\C^+\setminus (0, \I c]$ with simple zeros at points $\I\kappa_j$ of the discrete spectrum. It is continuous up to the boundary of the domain, with the only possible additional zero at the point $k=\I c$, the edge of the continuous spectrum. Unlike the case considered
in \cite{EGKT}, we admit the possible resonance at the point $\I c$, that is, we do not  assume the condition 
$
W(\I c)\neq 0
$
corresponding to the nonresonant case. In the resonant case the Wronskian has a square root zero at $k = \I c$ (cf. \cite{EGLT}). 

In $\mathcal V$ introduce also the function
\beq\label{defchi1} \chi(k):=\frac{4 k k_1}{ W(k) \breve W(k)}.  \eeq
From \eqref{soed2}  it follows that its limiting values satisfy
\beq\label{cont}[\chi(k)]_r=\I |\chi(k)|,\quad [\chi(k)]_l= -\I |\chi(k)|,\quad k\in[0,\I c].\eeq 
We also observe that
\beq\label{chinonres} \chi(k)=C (k-\I c)^{\ell/2}(1 + o(1)), \quad C\neq 0,\quad k\to \I c,\eeq where
\[ \ell:=\begin{cases} \begin{array} {rcl} 1, &\mbox{if}\quad W(\I c)\neq 0& (\mbox{nonresonant case});\\
-1,&\mbox{if}\quad W(\I c) = 0& (\mbox{resonant case}).\end{array}\end{cases}\]
Let $R(k)$ be the right reflection coefficient of the initial data satisfying  \eqref{decay} and let \[\gamma_j:=\|\phi(\I\kappa_j,\cdot,0)\|^{-2}_{L^2(\R)}\] be the  right normalizing constants for $j=1,...,N$. The set
\beq\label{scat5}\{ R(k),\; k\in \R; \quad |\chi(k)|, \ k\in[0, \I c];\,\quad  \I\kappa_j, \ \gamma_j,\ \ j=1,...,N\},
\eeq
constitute the minimal set of the  scattering data to reconstruct uniquely the solution of the initial value problem \eqref{kdv}--\eqref{ini} (cf. \cite{BET}, Corollary 4.4) 

Next, the Jost  solutions \eqref{lims} and \eqref{limsi} are connected by the  scattering relation
\beq\label{sct}
T(k,t) \phi_1(k,x,t) =  \ol{\phi(k,x,t)} +
R(k,t) \phi(k,x,t),  \qquad k \in \R,\eeq
where $T(k,t) =\frac{2 \I k}{W(k,t)}$ and $R(k,t)$  are the right  transmission and reflection coefficients. We use the notation $T(k) = T(k,0)$ and
$R(k)=R(k,0)$. Observe that \beq\label{deftr}|T(k)|^2=k\left[\frac{\chi(k)}{\sqrt{k^2 + c^2}}\right]_{r,l}, \quad k\in [0,\I c].\eeq

We define a vector-valued function $m(k,x,t) = (m_1(k,x,t), m_2(k,x,t))$, meromorphic in the spectral parameter $k \in \C \setminus(\R \cup [-\I c, \I c])$ for fixed $x,t$, as follows
\beq\label{defm}
m(k,x,t)= \left\{\begin{array}{cl}
\begin{pmatrix} T(k,t) \phi_1(k,x,t) \E^{\I kx},  & \phi(k,x,t) \E^{-\I kx} \end{pmatrix},
& k\in \C^+\setminus (0,\I c], \\
m(-k,x,t)\sigma_1,
& k\in\C^-\setminus [-\I c, 0),
\end{array}\right.
\eeq
where $\sigma_1=\begin{pmatrix}0&1\\1&0\end{pmatrix} $ is the first Pauli matrix. The vector function $m(k,x,t)$ evidently has at most simple poles at the points $\pm \I \kappa_j$. 
It is known that the following asymptotical formula, for $k\to \infty$, holds:
\[
m(k,x,t)= (1,1) -\frac{1}{2\I k}\left(\int_x^{+\infty}q(y,t)dy\right) (-1,1) + O\left(\frac{1}{k^2}\right).
\]
This expansion allows us to extract the shock wave solution using formula \eqref{previous}. However,  as was mentioned in the introduction, the formula \eqref{main1}, which can be computed using the well-known asymptotic formulas for the Weyl functions, is more convenient.

Indeed, it is known that for $k$ large enough   both functions $\phi(k,x,t)$ and $\phi_1(k,x,t)$ do not vanish  for all $x$ and $t$. Thus,
\begin{align*}
m_1(k,x,t)m_2(k,x,t)&=T(k,t)\phi(k,x,t)\phi_1(k,x,t)\\
&=\frac{2\I k}{\frac{\phi^\prime(k,x,t)}{\phi(k,x,t)} - \frac{\phi_1^\prime(k,x,t)}{\phi_1(k,x,t)}}=
\frac{2\I k}{\mathfrak m(k,x,t) -\mathfrak m_1(k,x,t)},
\end{align*}
where $\mathfrak m$ and $\mathfrak m_1$ are the right and left Weyl functions corresponding to the potential $q(x,t)$.  For $k\to\infty$ we have (cf. \cite{BM97}):
\[\aligned \mathfrak m(k,x,t)&=\I k +\frac{q(x,t)}{2\I k}  +\frac{f(x,t)}{4 k^2}+O(k^{-3}),\\ \mathfrak m_1(k,x,t)&=-\I k -\frac{q(x,t)}{2\I k}  +\frac{f(x,t)}{4 k^2} +O(k^{-3}).\endaligned\]
Thus,
\[ m_1(k,x,t)m_2(k,x,t) - 1=\frac{2\I k}{2\I k + \frac{q(x,t)}{\I k} + O(k^{-3}) } - 1=\frac{q(x,t)}{2k^2} + O(k^{-4}),\]
which proves \eqref{main1}.
\hfill $\square$

The following existence/uniqueness result is then valid:

\begin{theorem}\label{thm:vecrhp} 
Let \begin{itemize}\item the potential $q(x)$ satisfy \eqref{ini} and \eqref{decay};
\item the set \eqref{scat5} be
its right scattering data;
\item $\Sigma=\R\cup [\I c, -\I c]$ be the jump contour oriented left-to-right\,$\cup$\,top-down;
\item the phase function $\Phi(k)=\Phi(k,x,t)$ be defined by the formula: 
\end{itemize}\[
\Phi(k)= 4 \I k^3+\I k \frac {x}{t}, \quad k\in\C.
\] Then $m(k)=m(k,x,t)$ defined in \eqref{defm} for all $(x,t)\in\R\times\R_+$
is the unique solution of the following vector RHP:

\noindent Find a vector-valued function $m(k)$, meromorphic  away from $\Sigma$,  satisfying:
\begin{enumerate}
\item The jump condition $m_+(k)=m_-(k) v(k)$, where
\beq \label{eq:jumpcond}
v(k)=\left\{\begin{array}{cc}\begin{pmatrix}
1-|R(k)|^2 & - \ol{R(k)} \E^{-2t\Phi(k)} \\
R(k) \E^{2t\Phi(k)} & 1
\end{pmatrix},& k\in\R,\\
 \ &\ \\
\begin{pmatrix}
1 & 0 \\
\I|\chi(k)| \E^{2t\Phi(k)} & 1
\end{pmatrix},& k\in [\I c, 0],\\
 \ &\ \\
\sigma_1 (v(-k))^{-1}\sigma_1,& k\in [0, -\I c];\\
\end{array}\right.
\eeq
\item
the pole conditions
\beq
\aligned \label{eq:polecond}
\res_{\I\kappa_j} m(k) &= \lim_{k\to\I\kappa_j} m(k)
\begin{pmatrix} 0 & 0\\ \I \ga_j^2 \E^{t\Phi(\I \kappa_j)}  & 0 \end{pmatrix},\\
\res_{-\I\kappa_j} m(k) &= \lim_{k\to -\I\kappa_j} m(k)
\begin{pmatrix} 0 & - \I \ga_j^2 \E^{t\Phi(\I \kappa_j)} \\ 0 & 0 \end{pmatrix},
\endaligned
\eeq
\item
the symmetry condition
\beq \label{eq:symcond}
m(-k) = m(k) \sigma_1,\quad k\in \C\setminus\Sigma,
\eeq
\item
and the normalization condition
\beq\label{normco}
\lim_{\kappa\to\infty} m(\I\kappa) = (1\quad 1).
\eeq
\item In addition, the function $m(k)$ has the following behavior in a vicinity of the point  $ \I c$: If $\chi(k)$ satisfies \eqref{chinonres} with $\ell=1$ then $m(k)$ has continuous limits as $k$ approaches $\I c$ from the domain $\C\setminus\Sigma$. If  $\ell=-1$ then  one has 
\beq\label{resi} \aligned m(k)&=\left(C_1(x,t)(k-\I c)^{-1/2}, C_2(x,t)\right) (1+o(1)) \ \ \  C_1C_2\neq 0; \ \text{or}\ \\  m(k)&=\left( C(x,t), 0\right)(1 + o(1))\quad \mbox{as}\ \  k\to \I c.\endaligned \eeq
At the point $-\I c$ the analog of condition \eqref{resi} holds by symmetry \eqref{eq:symcond}.
\end{enumerate}
\end{theorem}

\begin{proof} The facts that $m$ satisfies the jump condition \eqref{eq:jumpcond} and the pole conditions \eqref{eq:polecond} are established in \cite{EGKT}. 
Note that the jump matrix on $\R$ also satisfies the symmetry $v(k)=\sigma_1 (v(-k))^{-1}\sigma_1$. To prove uniqueness,
assume first that $\tilde m(k)$ and $\hat m(k)$ are two solutions for the RHP {\it (i)--(v)}. Then  $\mu(k):=\tilde m(k) - \hat m(k)$  satisfies {\it (i)--(iii), (v)} and instead of {\it (iv)} we have 
\[ \mu(k)=O(k^{-1}), \quad k\to\infty.\]
In $\C^+\setminus (0, \I c]$ introduce the meromorphic function
\[F(k)=\mu_1(k)\ol{\mu_1(\ol k)} + \mu_2(k)\ol{\mu_2(\ol k)},\]
where $\mu_{1,2}$ are the components of $\mu$.
Then $F(k)=O(k^{-2})$ as $k\to\infty$.  Note that since the exact values of the  constants $C_1, C_2 $ and $C$ in \eqref{resi} are not specified, they may be different for $\ti m$ and $\hat m$. Furthermore, since $-\ol k=k$ for $k\in\I \R$, it follows from the symmetry condition {\it (iii)} that for such $k$,  $\mu_i(\ol k)=\mu_j(k)$, $i\neq j$. We thus get  $F(k)=O((k-\I c)^{-1/2})$ as $k\to \I c$ when $\ell=-1$. For the nonresonant case $\ell=1$ the function $F(k)$ has continuous limits everywhere on $\R\cup[0, \I c]$. Let us denote for simplicity $F_r(k)$ and $F_l(k)$ the limiting values of $F$ from the right and left sides of $[0, \I c]$, and $F_+(k)$ for the limiting values on the real axis from above. Then by the symmetry condition \eqref{eq:symcond} we get
\[\aligned
F_+(k)&=\mu_{1, +}(k)\ol{\mu_{1,-}(k)} + \mu_{2, +}(k)\ol{\mu_{2,-}(k)},\\
F_r(k)&=\mu_{1, r}(k)\ol{\mu_{2,l}(k)} + \mu_{2, r}(k)\ol{\mu_{1,l}(k)},\\
F_l(k)&=\mu_{1, l}(k)\ol{\mu_{2,r}(k)} + \mu_{2, l}(k)\ol{\mu_{1,r}(k)}.\endaligned\]
The jump condition \eqref{eq:jumpcond} implies
\[ F_+(k)=(1-|R(k)|^2)|\mu_{1,-}|^2 + |\mu_{2,-}|^2+
2\I\im\left(R(k)\E^{2t\Phi(k)}\ol{\mu_{1,-}(k)} \mu_{2,-}(k)\right),\]
\beq\label{314}\aligned F_l(k)&=\re\left(\mu_{1,l}(k)\ol{\mu_{2,l}(k)}\right) 
-\I|\mu_{2,l}(k)|^2 |\chi(k)|\E^{2t\Phi(k)},\\
F_r(k)&=\re\left(\mu_{1,l}(k)\ol{\mu_{2,l}(k)}\right) 
+\I|\mu_{2,l}(k)|^2|\chi(k)|\E^{2t\Phi(k)}.\endaligned\eeq
Note that $\Phi(k)\in\R$ for $k\in \I \R$.
From this and  \eqref{314} it follows that
\beq\label{315} \aligned\re F_l(k)&=\re F_r(k)= \re\left(\mu_{1,l}(k)\ol{\mu_{2,l}(k)}\right),\\ & \im F_l(k)=-\im F_r(k)\in\R_-.\endaligned\eeq
The pole condition \eqref{eq:polecond} is satisfied by the vector $\mu(k)$. Alongside with the symmetry property this implies 
\[\res_{\I\kappa_j} F(k)=2\I\gamma_j^2|\mu_2(\I\kappa_j)|^2\in\I\R_+.\]
Let now $\omega>c$ be arbitrary large and let $\mathcal C_\omega$ be the  boundary of the domain $\left(\C^+\cap\{k:\ \ |k|<\omega\}\right)\setminus (0,\I c]$. We treat $\mathcal C_\omega$ as a closed contour oriented counterclockwise. By Cauchy's theorem
\[\oint_{\mathcal C_\omega}F(k)dk=2\pi\I\sum_{j=1}^N \res_{\I\kappa_j} F(k),\]and since $F(k)=O(k^{-2})$ as $k\to\infty$,
the integral over the upper semicircle will asymptotically vanish as $\omega\to\infty$ and we get
\[\int_\R F_+(k)dk +\int_0^{\I c} F_l(k)dk + \int_{\I c}^0 F_r(k)dk +4\pi\sum_{j=1}^N \gamma_j^2|\mu_2(\I\kappa_j)|^2=0.\]
Taking into account \eqref{315}, the real part of this integral reads
\[\aligned 0 &=\int_\R \left((1-|R(k)|^2)|\mu_{1,-}|^2 + |\mu_{2,-}|^2\right) dk + 2\int_0^c |\mu_{2,l}(\I s)|^2 |\chi(\I s)|\E^{2t \Phi(\I s)}ds\\
&+4\pi\sum_{j=1}^N \gamma_j^2|\mu_2(\I\kappa_j)|^2.\endaligned\]
But $|R(k)|<1$ for $k\in\R\setminus\{0\}$, and therefore all summands in the last formula are non-negative. Thus, we obtain $\mu_2(\I\kappa_j)=0$ (which implies that $\mu_1(k)$ does not have a pole at $\I\kappa_j$) and
\[ \mu_{2,-}(k)=0,\ \ \mbox{for}\ \ k\in\R;\quad \mu_{2,l}(k) = 0, \ \ \mbox{for}\ \ k\in [\I c, 0]; \quad \mu_{1,-}(k)=0,\ \ \mbox{for}\ \ k\in\R.\]
From this and \eqref{eq:jumpcond} it immediately follows that $\mu_{1,+}(k) = \mu_{2,+}(k) = 0$ and $\mu_{2,r}(k)=\mu_{2,l}(k)=0$ for $k\in [\I c, 0]$. Thus, the function $\mu_{2}(k)$   is a holomorphic function in $\C$ with $\mu_2(k)\to 0$ as $k\to\infty$.  By Liouville's theorem $\mu_2(k)\equiv 0$ in $\C$. In turn, this identity and formula \eqref{eq:jumpcond} imply: $\mu_{1,r}(k)=\mu_{1,l}(k)$ for $k\in[\I c, 0]$. Therefore, $\mu_1(k)$ is also a holomorphic function in $\C$ vanishing at infinity, thus $\mu_1(k)\equiv 0$. This proves uniqueness. 

It remains to verify {\it (v)}. The case $\ell=1$ implies that the Wronskian $W(k,t)$ of the Jost solutions $\phi(k,x,t)$ and $\phi_1(k,x,t)$ does not vanish at $k=\I c$ for all $t$ (cf. \cite{EGT}, formula (6.2)).
 This implies that $T(k,t)$ is  bounded  and continuous as $k \to \I c$, and the same is true for the components of the vector $m$.

If $\ell=-1$ then  $W(\I c,t)=0$. Now, if, in turn $\phi(\I c, x,t)\neq 0$, then
$\phi_1(\I c, x, t)\neq 0$ (otherwise the Wronskian would not have zero at $k=\I c$).  This proves the first line of \eqref{resi}. If $\phi(\I c,x,t)=0$, then also $\phi_1(\I c,x,t)=0$.  Since $W(k,t)=\tilde C(t)(k-\I c)^{1/2}(1 + o(1))$ and $\phi_1(k,x,t)=\tilde C_1(x,t)(k-\I c)^{1/2}(1 + o(1))$ as $k\to \I c$, this proves the second line of \eqref{resi}.
\end{proof}
Theorem \ref{thm:vecrhp} guarantees the well-posedness of the initial meromorphic  vector RHP (IM RHP)
for all $(x,t)\in \R\times\R_+$. In the domain $\mathcal D_\varepsilon$ given by \eqref{DE}, \eqref{domain} where we intend to study and justify the asymptotics of its solution $m(k,x,t)$ as $k\to\infty$, the IM RHP admits an equivalent holomorphic statement.

\section {Holomorphic statement of the initial vector RHP}
\label{sec:inivrhp}

In this section we take a closer look at the ill-posedness of the associated matrix RHP. Let $a(\xi)$ be defined implicitly by \eqref{defa}, then as shown in \cite{EGKT}:
\[0<a(-\frac{c^2}{2}+\varepsilon)\leq a(\xi)\leq a(\frac{c^2}{3} -\varepsilon)<c.\]
Recall that the discrete spectrum is denoted by $-\kappa_j^2$, $j = 1, \dots , N$, with  $c<\kappa_1<...<\kappa_N$. Choose $\rho>0$ sufficiently small, such that 
\beq\label{defrhou}\aligned
\rho&<\frac{1}{4}\min\left\{\sqrt{c^2 + \eta^2}-c, \  \kappa_N - c,\ a(-\frac{c^2}{2}+\varepsilon),\right.\\ &\left. c- a(\frac{c^2}{3} -\varepsilon),\  \min\limits_{j=2,..,N}\,|\kappa_{j-1}-\kappa_j|,\ \rho_1\right\},\endaligned\eeq
 where $\eta>0$ is the decay estimate from \eqref{decay} and $\rho_1>0$ is defined implicitly by formula \eqref{rho1} below. Note, that since  $k_1 = \sqrt{c^2 + k^2}$ and $\eta>\sqrt{c^2 + \eta^2}-c$, the reflection coefficient $R(k)$ and the function $\chi(k)$ from \eqref{eq:jumpcond} are well defined in the domains 
\beq\label{Ommeg}\aligned\Omega_R&:=\{k: \rho>\im k>0\}, \ \mbox{ and}\\
 \ \Omega_\chi&:=\{k: \im k\in (0,c+\rho),\ |\re k|< \rho\}\setminus (0,\I c]\},\endaligned\eeq respectively, up to their boundaries.
 
Denote by \[\mathbb D_j:=\{k:\, |k-\I\kappa_j|<\rho.\}, \quad j = 1, \dots, N\]
and by
\[\mathbb T_j:= \partial \mathbb D_j =\{k:\ |k-\I\kappa_j|=\rho\},\quad j=1, \dots ,N\]
the small nonintersecting contours around the points of the discrete spectrum oriented counterclockwise.
Let $\mathcal C:=\{k:\, \im k=\rho\}$ be the upper boundary of $\Omega_R$ considered as a contour oriented from left to right.
We observe that with our choice of $\rho$ (cf. \eqref{defrhou})
\[\mathrm{dist}\, (\mathbb T_1,\, \ol{\Omega_\chi})>2\rho, \quad \mathrm{dist}\, (\I a\left(-c^2/2+\varepsilon\right),\ \ol{\Omega_R})>3\rho.\]
Introduce also the functions: 
\beq\label{Blaschke}P(k):=\prod_{j=1}^N \frac{k+\I\kappa_j}{k-\I\kappa_j},\quad k\in\C; \quad Q(k):=\left(\frac{k-\I c}{k+\I c}\right)^{\frac \ell 4},\quad k\in \C\setminus [-\I c, \I c ];\eeq
where $\ell$ is as in Theorem \ref{theormain} and $Q(\infty)=1$. 
 
Redefine now the solution $m(k)=m(k,x,t)$ of the IM RHP as follows:
\beq\label{defmini}
m^{\text{ini}}(k)=\begin{cases}
\begin{array}{ll} m(k)A_j(k)\left(P(k)Q(k)\right)^{-\sigma_3}, & k\in\mathbb D_j,\ \ j=1,..,N;\\[2mm]
m(k)A_0(k)\left(P(k)Q(k)\right)^{-\sigma_3}, & k\in \Omega_R;\\[2mm]
m(k)\left(P(k)Q(k)\right)^{-\sigma_3}, & k\in \C^+\setminus(\ol{\Omega_R}\cup [0,\I c]\cup\cup_{j=1}^N \ol{\mathbb D_j});\\[2mm]
m^{\text{ini}}(-k)\sigma_1,& k\in\C^-,\end{array}\end{cases}\eeq
where  we denoted \[
A_j(k)=\begin{pmatrix} 1& -\frac{k-\I\kappa_j}{\I\gamma_j^2\E^{2t\Phi(\I\kappa_j)}}\\0&1\end{pmatrix},\quad A_0(k)=\begin{pmatrix} 1&0\\ -R(k)\E^{t\Phi(k)} & 1\end{pmatrix}, \quad \sigma_3=\begin{pmatrix} 1&0\\0&-1\end{pmatrix}.\]
Introduce the contours in the lower half plane: $\mathcal C^*:=\{k:\ -k\in \mathcal C\}$ oriented right-to left, and $\mathbb T_j^*:=\{ k: -k\in \mathbb T_j\}$, $j=1,..., N$ oriented counterclockwise. Define the functions (cf. \eqref{Ommeg},\eqref{Blaschke}):
\beq\label{RX}\aligned\mathcal R(k)&:=R(k)P^{-2}(k)Q^{-2}(k),\quad  k\in\mathcal C;\\
X(k)&:=\chi(k)Q^{-2}(k)P^{-2}(k), \quad k\in \Omega_\chi.\endaligned\eeq Note that \beq\label{thesame}X_\pm(k)=\pm \I |X(k)|,\quad k\in [\I c, 0].\eeq Then we have the following
\begin{lemma}\label{leminiRH} For all $(x,t)\in \mathcal D_\varepsilon$ the vector function $m^{\text{ini}}(k)=m^{\text{ini}}(k,x,t)$ is the unique solution of the following RHP:

Find a vector-valued function, holomorphic in the domain  \[\C\setminus \Sini, \quad \Sini:=\mathcal C\cup\mathcal C^*\cup[\I c, -\I c]\cup \cup_j (\mathbb T_j\cup\mathbb T_j^*),\]
satisfying 
\begin{itemize}
\item the symmetry condition $m^{\text{ini}}(-k)=m^{\text{ini}}(k)\sigma_1$, $k \in \C \setminus \Sini $;
\item the jump condition $m^{\text{ini}}_+(k)= m^{\text{ini}}_-(k)v^{\text{ini}}(k)$, $k \in \Sini$, where
\beq\label{jumpini} v^{\text{ini}}(k)=\begin{cases}\begin{array}{ll}\begin{pmatrix} 1&0\\\mathcal R(k)\E^{2t\Phi(k)}& 1\end{pmatrix}, &k\in\mathcal C,\\[4 mm]
\begin{pmatrix} \exp(\frac{\I\ell\pi}{2}) &0\\ \I|X(k)|\E^{2t\Phi(k)}& \exp(\frac{-\I\ell\pi}{2})\end{pmatrix}, &k\in [\I c, \I \rho],\\[4 mm]
\begin{pmatrix} 1&\frac{(k-\I\kappa_j)P^2(k)Q^2(k)}{\I\gamma_j^2\E^{2t\Phi(\I\kappa_j)}}\\0&1\end{pmatrix}, &k\in\mathbb T_j,\\[2 mm]
\exp(\frac{\I\ell\pi}{2}\sigma_3) , & k\in [\I \rho, -\I\rho],\\[2 mm]
\sigma_1 [v^{\text{ini}}(-k)]^{-1}\sigma_1, & k\in [-\I\rho, -\I c],\\[2 mm]
\sigma_1 [v^{\text{ini}}(-k)]\sigma_1, & k\in \mathcal C^*\cup\cup_{j=1}^N\mathbb T_j^*.\end{array}\end{cases}\eeq
\item the normalizing condition $m^{\text{ini}}(k)\to (1,\ 1)$ as $k\to\infty$.
\item at points $\pm\I c$ it has at most fourth root singularities:

\noindent $m^{\text{ini}}(k)=O(k\mp \I c)^{-1/4}$ as $k\to\pm \I c$. 
\end{itemize}
\end{lemma} 
\begin{figure}[h]
\vspace{35pt}
\begin{picture}(8,7)

\put(4,0.2){\line(0,1){0.8}}
\put(4,3){\line(0,1){0.8}}
\put(4,1){\line(0,1){2}}
\put(0,2.4){\line(1,0){8}}
\put(0,1.6){\line(1,0){8}}

\put(4,3.25){\vector(0,-1){0.1}}
\put(4,2.05){\vector(0,-1){0.1}}
\put(4,0.75){\vector(0,-1){0.1}}

\put(2.1,2.4){\vector(1,0){0.1}}
\put(2.1,1.6){\vector(-1,0){0.1}}
\put(6.1,2.4){\vector(1,0){0.1}}
\put(6.1,1.6){\vector(-1,0){0.1}}

\put(4,5.4){\vector(1,0){0.1}}
\put(4,7.1){\vector(1,0){0.1}}

\put(4,-3.1){\vector(-1,0){0.1}}
\put(4,-1.4){\vector(-1,0){0.1}}

\put(4.2,2.5){$\I \rho$}
\put(4.05,1.3){$-\I \rho$}
\put(4,1.6){\circle*{0.1}}
\put(4,2.4){\circle*{0.1}}

\put(4.2,3.7){$\I c$}
\put(4.05,0.2){$-\I c$}
\put(4,0.2){\circle*{0.1}}
\put(4,3.8){\circle*{0.1}}

\put(3.85,5.95){$\kappa_1$}
\put(3.6,-2.1){$-\kappa_1$}
\put(4,-1.8){\circle*{0.1}}
\put(4,5.8){\circle*{0.1}}
\put(4,-1.8){\circle{0.8}}
\put(4,5.8){\circle{0.8}}
\put(3.95,6.5){$\vdots$}

\put(3.85,7.65){$\kappa_N$}
\put(3.6,-3.8){$-\kappa_N$}
\put(4,-3.5){\circle*{0.1}}
\put(4,7.5){\circle*{0.1}}
\put(4,-3.5){\circle{0.8}}
\put(4,7.5){\circle{0.8}}
\put(3.95,-2.8){$\vdots$}

\put(3,6){$\mathbb T_1$}
\put(3,7.7){$\mathbb T_N$}
\put(3,-1.6){$\mathbb T_1^*$}
\put(3,-3.3){$\mathbb T_N^*$}

\put(4,0.2){\circle*{0.1}}
\put(4,3.8){\circle*{0.1}}

\put(2,2.6){$\mathcal C$}
\put(2,1.15){$\mathcal C^*$}

\put(8.2,1.9){$\mathbb R$}

\curvedashes{0.05,0.05}
\curve(0,2, 8,2)
\end{picture}
\vspace{110pt}
\caption{Jump contour $\Sini$}
\end{figure}

\begin{proof} The proof is very similar to that one given in \cite{EGKT}, with only one difference:
we can use the identity $R_-(k)-R_+(k) +\I |\chi(k)|=0$ (Lemma 3.2, \cite{EGKT}) on the interval $[\I\rho, 0]$ and take into account the influence of the function $Q(k)$. In particular, we used that $Q_-(k)Q_+^{-1}(k)=\exp(\frac{\I\ell\pi}{2})$, and $Q_+(k)Q_-(k)=|Q^2(k)|$ for $k\in [\I c, -\I c]$.\end{proof}
Since $m^{\text{ini}}(k)$ is a piecewise holomorphic vector function, we call the problem stated in Lemma \ref{leminiRH} the initial holomorphic (IH) RHP. As already mentioned in the Introduction, the symmetry condition is crucial for uniqueness and plays an essential role in the final asymptotical analysis.  
That is why, all transformations steps carried out  from the initial RHP to an RHP asymptotically close to an exactly solvable model vector RHP, should respect the following {\it symmetry conditions}:  

\begin{hypothesis}\label{remsym}  Each vector RHP  should satisfy:
\begin{itemize}\item The jump contour $\Sigma$ is symmetric with respect to the map $k\mapsto -k$;
\item On $\C\setminus \Sigma$ the vector solution $m(k)$ is holomorphic and satisfies

\noindent $m(-k)=m(k)\sigma_1$;
\item Let $\mathcal L\subset \Sigma$ be a subcontour of $\Sigma$ which does not contain symmetric points. We denote by $\mathcal L^*=\{k:\, -k\in \mathcal L\} \subset \Sigma$ its inversion, if $\mathcal L^*$ has the orientation of the following type: when $k$ moves in the positive direction along $\mathcal L$, then
$-k$ moves in the positive direction along $\mathcal L^*$. In this case, the jump matrix $v(k)$ of the jump problem \beq\label{jc}m_+(k)=m_-(k)v(k),\quad k\in\Sigma,\eeq should satisfy $\det v(k)=1$ and the symmetry
\beq\label{symv1}v(-k)=\si_1 v(k)\si_1,\quad k\in\mathcal L\cup\mathcal L^*.\eeq
If the inversion of $\mathcal L$ has the opposite orientation, we denote it by $(\mathcal L^*)^{-1}$.
For example,  $\mathcal L=[\I c, 0]$ and $(\mathcal L^*)^{-1}=[0, -\I c]$ are both oriented top-bottom.
In this case, 
\beq\label{symv2}v(-k)=\si_1 v(k)^{-1}\si_1, \quad k\in \mathcal L\cup (\mathcal L^*)^{-1};\eeq
\item The vector-function $m(k)$ is continuous up to the boundary, except at the node points of the contour (the ends and self intersections of $\Sigma$, and a finite number of points of discontinuity of the jump matrix), where fourth root singularities are admissible;
\item $m(k)\to (1, 1)$ as $k\to \infty$.
\end{itemize}
\end{hypothesis}

Evidently, the IH RHP formulated in Lemma \ref{leminiRH} satisfies all these requirements.
Alongside with it, we can write down an analogous matrix RHP with the same jump matrix  $v^{\text{ini}}(k)$ given by  \eqref{jumpini}. This can be done in two ways. Either by imposing a symmetry condition on the matrix solution(see \eqref{uniq3} below), or by the standard normalization to the unit matrix $\id$ at infinity.   Simultaneous use of both conditions may seem excessive. In  fact, we observe the following. 

Let $\Sigma\subset\C$ be a union of finitely many smooth curves (finite or infinite) which intersect in at most a finite number of points and all intersections are transversal (this condition can of course be relaxed, but it is sufficient for the applications we have in mind). We will also require $\Sigma$ to be symmetric  with respect to the inversion $k\mapsto -k$.   

Let now $v(k)$ be a piecewise continuous bounded matrix function on  $\Sigma$ satisfying \eqref{symv1} or \eqref{symv2}, with $\det v(k)\equiv1$.  The points of discontinuity of the jump matrix, together with the (finite) set of boundary points $\partial \Sigma$ and the self intersection points of $\Sigma$, are denoted by $\mathcal G$.  We assume that $0\notin \mathcal G$.

Finally, let $\mathcal H$ be the class of $2\times 2$ matrix functions $M(k)$ holomorphic in $\C\setminus\Sigma$, which have continuous limits up to the boundary $\Sigma\setminus\mathcal G$, and have a limit as $k\to\infty$ avoiding $\Sigma$. At  points of $\mathcal G$ we allow  singularities of the form:   
\beq\label{sing}
M(k)=O((k-\kappa)^{-1/4}),\quad \mbox{as}\ \ k\to\kappa\in\mathcal G.
\eeq
Now for an admissible $M\in \mathcal{H}(\Sigma)$ we consider the following RHP
\beq\label{uniq1}
M_+(k)=M_-(k)v(k),\qquad k\in\Sigma,
\eeq
together with the  normalization condition
\beq\label{uniq2}
M(\infty):=\lim_{k\to\infty} M(k)=\id
\eeq
and  the symmetry condition
\beq\label{uniq3}
M(-k)=\sigma_1 M(k)\sigma_1, \qquad k\in\C\setminus\Sigma.
\eeq

\begin{theorem} \label{theor1}
Suppose $\Sigma\subset\C$ is an admissible contour and $v(k)$, $k\in\Sigma$ an admissible matrix as specified above.
Then the following propositions are valid:
\begin{enumerate} [(a)]
\item
If a solution $M\in \mathcal{H}(\Sigma)$ of the jump problem \eqref{uniq1} exists and $\det M(\infty)\neq 0$, then the matrix $M(\infty)^{-1}M(k)$ solves
the RHP \eqref{uniq1}--\eqref{uniq2}, and every other solution of  \eqref{uniq1} is given by $\ti{M}(k) = \ti{M}(\infty) M(\infty)^{-1}M(k)$ in this case.
Moreover,  $\det M(k) = \det M(\infty)$.
\item
If the jump problem \eqref{uniq1} has a nonsingular, that is invertible solution from $\mathcal H(\Sigma)$, then every solution $M\in \mathcal{H}(\Sigma)$ of \eqref{uniq1}
satisfies the symmetry condition \eqref{uniq3} provided $M(\infty)$ satisfies the symmetry condition. In this case $M$ is of the form
\[
M(k)= \begin{pmatrix}\alpha(k) & \beta(k)\\ \beta(-k) & \alpha(-k)\end{pmatrix}, \qquad M(\infty)= \begin{pmatrix} a& b\\b & a\end{pmatrix}
\]
with $\det M(\infty)=a^2-b^2$.  If $M$ is nonsingular then  $a+b\neq 0$.
\item
Suppose \eqref{uniq1} has a nonsingular solution $M$ satisfying \eqref{uniq3}. Then  the vector function
$m(k)$  
\[
m(k)=\frac{1}{a+b} (1,\ 1) M(k) = \frac{1}{a+b} (\alpha(k)+\beta(-k),\, \beta(k)+\alpha(-k)).
\]
solves the same jump problem $m_+(k)=m_-(k) v(k)$ and satisfies  \eqref{eq:symcond} and \eqref{normco}.
 Moreover, in this case $m$ is the unique solution of this problem with admissible singularities of the type \eqref{sing}.
\item
Suppose the vector problem satisfying Hypothesis \ref{remsym} has a solution $m$ which satisfies the condition $m_\pm(0)=(0,0)$.
Then there is no invertible solution of the problem \eqref{uniq1}, \eqref{uniq3} in $\mathcal H (\Sigma)$.
\end{enumerate}
\end{theorem}

\begin{proof}
 {\it (a)}. This follows similarly as in \cite[Theorem~7.18]{deiftbook}. 
 
 \noindent
{\it(b)}. Let $M(k)\in\mathcal H(\Sigma)$ be the solution of the problem \eqref{uniq1}--\eqref{uniq2}. By {\it(a)} it suffices to show that $M$ satisfies \eqref{uniq3}.
To this end set $\tilde M(k)=\sigma_1 M(-k)\sigma_1$. Then $\tilde M(\infty)=\id$ and $\tilde M(k)\in\mathcal H$. Taking into account the symmetry of $\Sigma$, for example, \eqref{symv2}, we see that
\begin{align*}
 \tilde M_+(k)&=\sigma_1 M_-(-k)\sigma_1=\sigma_1 M_+(-k)v^{-1}(-k)\sigma_1\\
&=\sigma_1 M_+(-k)\sigma_1 \sigma_1 v^{-1}(-k)\sigma_1=\tilde M_-(k) v(k).
\end{align*}
Thus $\tilde M(k)$ solves \eqref{uniq1}--\eqref{uniq2} and by uniqueness, $\tilde M(k)\equiv M(k)$. This proves \eqref{uniq3}. The rest is straightforward.

\noindent {\it(c).} By assumption we have a solution $M$ as in {\it(b)} and hence one easily checks that $m$ satisfies \eqref{jc}, as well as \eqref{eq:symcond} and \eqref{normco}. If $\ti{m}$ is a second solution, then as in {\it(a)}
we see that \eqref{jc} implies that $c=\ti{m}(k)M^{-1}(k)$ is a constant vector. Hence by \eqref{normco} we see $c=\frac{1}{a+b} (1, 1)$.

\noindent{\it(d)}. Suppose that there exists an invertible symmetric matrix $M(k)$ satisfying \eqref{uniq1}. Without loss of generality we can assume $M(\infty)=\id$ and
hence by {\it(c)} this implies $m_+(0)=(\alpha_+(0) + \beta_-(0),\beta_+(0) + \alpha_-(0))=(0,0)$. Consequently
\[
M_+(0)=\begin{pmatrix} \alpha_+(0) & \beta_+(0)\\ -\alpha_+(0)& -\beta_+(0)\end{pmatrix}
\]
implying $\det M(k) = \det M_+(0)=0$.
\end{proof}

In particular, item {\it(d)} implies that any technique relying on existence of a bounded nonsingular matrix solution is bound to fail at all points in the $(x,t)$-plane where $m_+(0,x,t)=(0,0)$
holds.  Recall now that the vector function \eqref{defm} is the unique solution of the IM RHP, making $\mini$ the unique solution of the IH RHP. After the transformation \eqref{defmini} the point $k=0$ became an inner point of the contour $\Sini$.  Moreover, taking into account the scattering relation \eqref{sct} and the fact $\phi(+0,x,t)=\phi(-0,x,t)=\ol{\phi(+0,x,t)}$, it is straightforward to check that
\[\aligned m^{\text{ini}}_\pm(0,x,t)&=\left(\ol{\phi(\pm 0,x,t)}P^{-1}(0)Q_\pm^{-1}(0),\ \phi(\pm 0,x,t)P(0)Q_\pm(0)\right)\\
&=(-1)^N\phi(0,x,t)\left(\E^{\pm\frac{\I\ell\pi}{4}},\ \E^{\mp\frac{\I\ell\pi}{4}}\right)\endaligned\] 
Thus, if $\phi(0,x^*,t^*)=0$ for  arbitrary large $(x^*, t^*)\in\mathcal D_\varepsilon$, then 
 $m^{\text{ini}}_\pm(0,x^*,t^*)=0$ and by Theorem \ref{theor1}, {\it (d)} we can talk about ill-posedness of the respective matrix RHP.
Moreover, even for the one-soliton (reflectionless, decaying) case this occurs as pointed out in the discussion after Lemma~2.5 in \cite{GT}.

\section{From the IH RHP to the model RHP}\label{sec4}

Now we recall briefly the conjugation and deformation steps which lead to the model problem solution in the domain  $\mathcal D_\varepsilon$.  As is shown in \cite{EGKT} (see also \cite{EGT}), for $\xi=\frac{x}{12 t}\in (-\frac{c^2}{2}, \, \frac{c^2}{3})$ the equality \eqref{defa}
generates an implicitly given positive smooth  function $a(\xi)$, monotonously increasing such that $a(-\frac{c^2}{2})=0$,
 $a(\frac{c^2}{3})=c$.  
  In the domain $\C\setminus [\I c, -\I c]$ we introduce the function
\beq\label{deffung}
{g(k):=g(k,x,t)=12\int_{\I c}^k \left(k^2 +  \xi +\frac{c^2-a^2}{2}\right)\sqrt\frac{k^2 + a^2}{k^2 + c^2}} dk.
\eeq
Here we use the standard branch of the square root with the cut along $\R_-$. 
\begin{lemma}\label{lemg} (\cite{EGKT}).
The function $g$ posseses the following properties
\begin{enumerate}[{\bf (a)}]
\item  $g(k)=-g(-k)$ for $k\in \C\setminus [\I c, - \I c]$;
    \item $g_-(k)+g_+(k)=0$ as $k\in [\I c, \I a]\cup [-\I a, - \I c]$;
\item $g_-(k) - g_+(k)=B$ as $k\in [\I a, - \I a]$, where $ B:=B(\xi)= -2g_+(\I a)>0$;
\item the asymptotical behavior
\[\Phi(k,\xi)-\I g(k,\xi)= O\left(\frac{1}{k}\right).\]
 holds as $k\to\infty$.
\end{enumerate}
\end{lemma}
The signature table for the imaginary part of function $g$ is shown in the following figure:
\begin{figure}[h]
\begin{picture}(8,4)
\put(0,2){\vector(1,0){8}}
\put(4,0){\line(0,1){1}}
\put(4,3){\line(0,1){1}}

\put(0.5,2.3){$+$}
\put(0.5,1.6){$-$}
\put(7.5,2.3){$+$}
\put(7.5,1.6){$-$}
\put(3.6,3.9){$-$}
\put(3.6,0){$+$}

\put(4.2,2.9){$\I a$}
\put(4.2,1){$-\I a$}
\put(4,1){\circle*{0.1}}
\put(4,3){\circle*{0.1}}

\put(4.2,3.9){$\I c$}
\put(4.2,0){$-\I c$}
\put(4,0){\circle*{0.1}}
\put(4,4){\circle*{0.1}}

\curve(4,3, 5.5,2.5, 7.5,4)
\curve(4,1, 5.5,1.5, 7.5,0)
\curve(4,3, 2.5,2.5, 0.5,4)
\curve(4,1, 2.5,1.5, 0.5,0)

\curvedashes{0.05,0.05}
\curve(4,1, 4,3)

\end{picture}
  \caption{Sign of $\im(g)$}\label{fig1}
\end{figure}

STEP 1. Let $m^{\text{ini}}(k)$ be the unique vector solution of the IH RHP. Redefine it by \beq\label{firsttr}m^{(1)}(k):=m^{\text{ini}}(k)\E^{(\I tg(k) - t\Phi(k))\sigma_3}.\eeq
Then $m^{(1)}(k)$ is a piecewise-holomorphic function  in $\C$ which satisfies the symmetry requirements of Hypothesis \ref{remsym} and solves the jump problem $m^{(1)}_+(k)= m^{(1)}_- (k)v^{(1)}(k)$ with
\[
v^{(1)}(k) =  \begin{cases}\begin{array}{ll}\begin{pmatrix} 1&0\\ \mathcal R(k)\E^{2\I
tg(k)}& 1\end{pmatrix}, &k\in\mathcal C,\\[4 mm]
\begin{pmatrix} \E^{\I t (g_+ -g_-) +\frac{\I \ell\pi}{2}} &0\\ \I |X(k)|\E^{\I t(g_+ +g_-)}& \E^{-\I t (g_+ -g_-) -\frac{\I\ell\pi}{2}}\end{pmatrix}, &k\in [\I c, \I \rho],\\[4 mm]
\begin{pmatrix} 1& h_j(k,t)\\0&1\end{pmatrix}, &k\in\mathbb T_j,\ j=1,..,N,\\[2 mm]
\exp((-\I tB+\I\frac{\ell\pi}{2})\sigma_3) , & k\in [\I \rho, -\I\rho],\\[2 mm]
\sigma_1 [v^{(1)}(-k)]^{-1}\sigma_1, & k\in [-\I\rho, -\I c],\\[2 mm]
\sigma_1 [v^{(1)}(-k)]\sigma_1, & k\in \mathcal C^*\cup\cup_j\mathbb T_j^*,\end{array}\end{cases}
\]
where $\mathcal R(k)$ and $X(k)$ are given by \eqref{RX} and 
\[h_j(k,t):=h_j(k,t,\xi)=-\I(k-\I\kappa_j)P^2(k)Q^2(k)\gamma_j^{-2}\E^{-2t\left(\Phi(\I\kappa_j)-\Phi(k)\right)-2t\I g(k)}.\]
Since $\im g(\I \kappa_j)< -\delta <0$ (cf. Figure \ref{fig1}), uniformly with respect to $\xi\in \mathcal I_\varepsilon$, we conclude that there exists $\rho_1>0$ such that
\beq\label{rho1}
\max_{j=1,..,N}\sup_{ |k-\kappa_j|\leq \rho_1}\left(|\Phi(\I \kappa_j)-\Phi(k)|+\im g(k)\right)< -C(\varepsilon)<0.\eeq
Taking into account \eqref{defrhou}, we prove
\begin{lemma}\label{refin}
The following estimate is valid uniformly with respect to $\xi\in \mathcal I_\varepsilon$ \footnote{by $C(\varepsilon)$ we will denote any positive constant with respect to $k$, $\xi$, $x$ and $t$}
\[\max_j\sup_{k\in \mathbb T_j}|h_j(k,t)|= O(\E^{-C(\varepsilon)t}).\]
\end{lemma}
Put now $b:=a-\rho$. Recall that the smoothness of the initial data \eqref{decay} up to the 7-th derivative implies that $R(k)$ for $k \in \R$ is a smooth function with $R(k)=O(k^{-6})$ as $k\to\pm\infty$ (\cite{EGLT}, Theorem 4.1). From item {\bf (c)} of Lemma \ref{lemg}, \eqref{defrhou} and the signature table of $g(k)$ we conclude that the following proposition is valid:
\begin{lemma}\label{lilem} Uniformly with respect to $\xi\in\mathcal I_\varepsilon$
\[\|v^{(1)}(k) -\id\|_{L^{\infty}(\mathcal C)} + \|v^{(1)}(k) -\id\|_{L^1(\mathcal C)}=O(\E^{-C(\varepsilon)t}),\]
\[\| v^{(1)}(k) -\E^{\left(-\I t B +\I\frac{\pi\ell}{2}\right)\sigma_3}\|_{L^{\infty}([0,\I b])}=O(\E^{-C(\varepsilon)t}).\]
\end{lemma}

STEP 2. 
Our next  conjugation step deals with a factorization of the jump matrix on the set $[\I c, \I a]\cup [-\I a, -\I c]$. To this end consider the following function $F(k)=F(k,\xi)$, $k\in\C\setminus [\I c, -\I c]$:
\beq\label{forF}
F(k):=\exp\left\{\frac{w(k)}{2\pi\I}\left(\int_{\I c}^{\I a}\frac{f(s)}{s-k}ds +\int_{-\I c}^{-\I a}\frac{f(s)}{s-k}ds -\I\Delta_F \int_{-\I a}^{\I a}\frac{ds}{w(s)(s-k)}\right)\right\},
\eeq
where 
\[
w(k)=\sqrt{(k^2+c^2)(k^2 + a^2)},\quad k\in \C\setminus ([\I c, \I a]\cup [-\I a, -\I c]),\quad w(0)>0,
\]
\beq\label{deff}f(k):=\frac{\log |X(k)|}{w_+(k)}, \quad k\in [\I c, -\I c],
\eeq
and 

\beq \label{defdelta}
\Delta_F=\Delta_F(\xi):=2\I\int_{\I a}^{\I c}f(s)ds\left(\int_{-\I a}^{\I a}\frac{ds}{w(s)}\right)^{-1}\in\R.\ \eeq
\begin{remark} Putting together \eqref{deftr}, \eqref{RX} and \eqref{Blaschke} we conclude that $\Delta= \Delta(\xi)$ given by \eqref{deltaini} and $\Delta_F$ given by \eqref{defdelta}, \eqref{deff} are connected by 
\beq\label{svyazF} \Delta=\Delta_F -\frac{\ell\pi}{2}.\eeq
\end{remark}
Next, since $X(k)=\chi(k)Q^{-2}(k)P^{-2}(k)$ has bounded non-vanishing values at points $\pm \I c$, we get  
\begin{lemma}[\cite{EGKT},\cite{MUSK}]\label{lemF} The function $F(k)$ possesses the following properties:
\begin{enumerate} [ (1) ]
\item  $F(-k)=F^{-1}(k)$ for $k\in\C\setminus [\I c, -\I c]$;
\item $F_+(k)F_-(k)=|X(k)|$ for $k\in[\I c,\I a]$;
\item $F_+(k)=F_-(k)\E^{\I\Delta_F} $ for $k\in[\I a, - \I a]$;
\item $F(k)\to 1$ as $k\to \infty$; 
\item $F_+(k)F_-(k)=\left(F_+(-k)F_-(-k)\right)^{-1}$ for $k\in [-\I a, -\I c]$;
\item $F(k)$ has finite limits as $k\to\pm\I c$.
\end{enumerate}
\end{lemma}
Taking into account these properties and property \eqref{cont} we observe that the matrix $v^{(1)}(k)$ can be factorized on $[\I c, \I a]$ as follows:
\[
v^{(1)}(k)=  G_-(k)\begin{pmatrix}0&\I\\ \I & 0\end{pmatrix} G_+(k)^{-1},\]
where (cf. \eqref{Ommeg}, \eqref{RX}, \eqref{Blaschke}, \eqref{defchi1}):
\[ G(k)= \begin{pmatrix} F^{-1}(k) & -\frac{ F(k)   \E^{-2\I t g(k)}}{X(k)}\\[2mm] 0 & F(k) \end{pmatrix},\quad k\in \Omega_\chi.
\]
Inside the domain $\Omega_\chi$ introduce the subdomain $\Omega_1$ surrounded by the contour $\Sigma_1$ oriented  as depicted
 in Figure~\ref{fig3}. Denote by $\Sigma_1^*$ its inversion in $\C^-$.
Define $m^{(2)}(k)$ as
\[
m^{(2)}(k):=m^{(1)}(k) \left\{\begin{array}{ll} G(k), & k\in \Omega_{1},\\
(F(k))^{-\sigma_3},  & k\in \ol{\C^+}\setminus\ol{\Omega_{1}},\\
m^{(2)}(-k)\si_1, &k\in \C^-.\end{array}\right.
\]
Since $F(k)\to 1$ as $k\to \infty$, the normalization condition is preserved for $m^{(2)}(k)$. The correctness of its definition by symmetry in the lower half plane is due to properties of {\it (1), (2), (5)} of Lemma \ref{lemF}. Moreover, due to property {\it(6)}, \eqref{firsttr} and Lemma \ref{leminiRH}, we have
\begin{align*}
m^{(2)}(k)=O(k\mp \I c)^{-1/4}, \ \mbox{as}\ &k\to\pm \I c; \quad m^{(2)}(k)=O(1), \ \mbox{as}\ k\to\pm \I a;
\\
m^{(2)}(k)&=O(1), \ \mbox{as}\ k\to\pm \I \rho.
\end{align*}
Note that the set $\mathcal G^{(2)}=\{\pm\I c, \pm\I a, \pm\I \rho\}$ is the set of all node points of the RHP for $m^{(2)}(k)$.
Taking into account property ${\bf (c)}$ of Lemma \ref{lemg}, property $(3)$ of Lemma \ref{lemF} and \eqref{svyazF},  we see that
 \beq\label{propimp4} \frac{F_-(k)}{F_+(k)}\E^{\I t(g_+(k) - g_-(k)+\I\ell\pi/2) }= \E^{-\I t  B-\I \Delta}, \quad k\in [\I a, -\I a],\eeq
and therefore the jump matrix for $m^{(2)}(k)$ looks as follows \beq\label{jumpcond25}
v^{(2)}(k) = \left\{ \begin{array}{ll}
\begin{pmatrix}
0 & \I \\
\I  & 0
\end{pmatrix},& k\in [\I c, \I a],\\ [3mm]
\begin{pmatrix} \frac{F_-}{F_+}\E^{\I t(g_+ - g_-) +\I\ell\pi/2 }& 0\\ 
\frac{\I |X|}{F_+F_-} \E^{\I t(g_+ + g_-)) } & \frac{F_+)}{F_-}\E^{\I t(g_- - g_+)-\I\ell\pi/2 }\end{pmatrix} & k\in [\I a, \I b]\\ [5mm]
\begin{pmatrix}\E^{-\I t  B-\I \Delta}& 0\\
0&\E^{\I t  B+\I\Delta}\end{pmatrix} +\mathcal A(k,t),&  k\in [\I b, 0],\\ [3mm]
\begin{pmatrix} 1 & -\frac{F^2(k)}{X(k)}\E^{-2\I t g(k)}\\ 0 & 1\end{pmatrix}, & k\in\Sigma_1,\\ [3 mm]
\begin{pmatrix} 1 & 0\\ \mathcal R(k)F^{-2}(k)\E^{2\I t g(k)} & 1\end{pmatrix}, & k\in \mathcal C,\\ [3mm]
\begin{pmatrix} 1& F^2(k)h_j(k,t)\\0&1\end{pmatrix}, &k\in\mathbb T_j,\ j=1,..,N,\\[2 mm]
\sigma_1 (v^{(2)}(-k))^{-1}\sigma_1, & k\in  [0, -\I c]\\
\sigma_1 v^{(2)}(-k)\sigma_1, & k\in \Sigma_1^*\cup\mathcal C^*\cup \cup_j\mathbb T_j^*,
\end{array}\right.
\eeq
where the matrix \[ \mathcal A(k,t)=[F_-(k)]^{\si_3}\left( v^{(1)}(k)-\E^{(-\I t B+\I\frac{\ell\pi}{2})\si_3}\right)[F_-(k)]^{-\si_3}\] is supported on $[\I b, \I\rho]$ and
admits, according to Lemma \ref{lilem}, the estimate
\[
\|\mathcal A(k,t)\|_{L^{\infty}([0, \I b])}=O(\E^{-C(\varepsilon)t}).\]
\begin{figure}[H]
\begin{picture}(8,7)

\put(4,5.5){\line(0,-1){3.8}}
\put(0,2.7){\line(1,0){8}}

\curve(4,4.6, 3.2,5.7, 4,6.3, 4.8,5.7, 4,4.6)


\put(4,3.15){\vector(0,-1){0.1}}
\put(4,2.15){\vector(0,-1){0.1}}

\put(4,4.15){\circle*{0.1}}

\put(4,5.1){\vector(0,-1){0.1}}

\put(2.1,2.7){\vector(1,0){0.1}}
\put(6.1,2.7){\vector(1,0){0.1}}


\put(4.5,5.1){\vector(1,1){0.1}}

\put(4.2,2.4){$\I \rho$}

\put(4,2.7){\circle*{0.1}}

\put(4.1,5.5){$\I c$}
\put(4,5.5){\circle*{0.1}}

\put(4.2,4.05){$\I b$}

\put(4.2,4.5){$\I a$}
\put(4,4.6){\circle*{0.1}}


\put(4,1.7){\circle*{0.1}}


\put(8.2,1.6){$\mathbb R$}
\put (4, 1.3){$0$}
\put(3.5,5.7){$\Omega_1$}

\put(3.15,4.8){$\Sigma_1$}

\curvedashes{0.05,0.05}

\curve(0,1.7, 8,1.7)

\curve(3,1.7, 3,6.5)
\curve(3,6.5, 5,6.5)
\curve(5,1.7, 5,6.5)
\put(3.3,3.4){$\Omega_\chi$}

\put(2,2.9){$\mathcal C$}
\put(2,2.05){$\Omega_R$}
\end{picture}
\vspace{-40pt}
\caption{Jump contour $\Sigma^{(2)}$ in $\C^+$ (without the $\mathbb T_j$'s)}\label{fig3}
\end{figure}
Lemma \ref{lilem} and Lemma \ref{refin} together with properties {\it (1)} and {\it (4)} of Lemma \ref{lemF} also imply
\begin{lemma}\label{estid} Uniformly with respect to $\xi\in \mathcal I_\varepsilon$
\[\|v^{(2)}(k) -\id\|_{L^{\infty}(\mathcal K)} + \|v^{(2)}(k) -\id\|_{L^1(\mathcal K)}=O(\E^{-C(\varepsilon)t}),\ \mbox{as}\ \ t\to\infty,\]
where
$\mathcal K=\mathcal C\cup\mathcal C^*\cup\cup_j(\mathbb T_j\cup\mathbb T_j^*).$\end{lemma}
\begin{remark}\label{remmod} Formula  \eqref{propimp4} 
allows us to shorten the expression for $v^{(2)}(k)$ on the interval  $[\I a, \I b].$ However, we use the form \eqref{jumpcond25} of the jump matrix on $[\I a, \I b]\cup [-\I b, -\I a]$,  because it simplifies further considerations of the  local parametrix problem. \end{remark}

Let $\mathcal B$ be a vicinity of point $\I a$ with the boundary $\pa\mathcal B$ satisfying
\[ \frac{\rho}{2}<\mbox{dist}\,( \pa\mathcal B, \I a)<2\rho,\]
where $\rho$ is defined by \eqref{defrhou}. Its precise shape will be described later in Section~\ref{sec7}. Without loss generality one can assume that $\I b\in\pa\mathcal B$. Denote $\mathcal B^*=\{k:\  -k\in\mathcal B\}$ and the jump contour for  $m^{(2)}(k)$ by \beq\label{hatsi}
\Sigma^{(2)}:=\mathcal C\cup\mathcal C^*\cup \Sigma_1\cup \Sigma_1^*\cup \cup_j(\mathbb T_j\cup\mathbb T_j^*)\cup [\I c,  -\I c],\eeq
and let \beq\label{estmoro}\Sigma_\rho=\Sigma^{(2)}\setminus (\mathcal B\cup \mathcal B^*)\eeq be the part of our contour outside the small vicinities of the points $\pm\I a$.
 Put
 \beq\label{vm}
v^{\text{mod}}(k) = \left\{ \begin{array}{ll}
\I\sigma_1,& k\in [\I c, \I a], \,\\
\E^{-\I \Lambda\sigma_3},& k\in [\I a, 0],\\
\si_1(v^{\text{mod}}(-k))^{-1}\si_1, & k\in [0, -\I c],\\
\id, & k\in\Sigma^{(2)}\setminus [\I c, -\I c],
\end{array}\right.
\eeq
where
\beq\label{La} 
\Lambda:=t B +\Delta\in\R.
\eeq
The consideration above shows that
 uniformly with respect to $\xi\in\mathcal I_\varepsilon$
 \beq\label{estmoem}
 \| v^{(2)}(k) - v^{\text{mod}}(k)\|_{L^\infty(\Sigma_\rho)\cap L^1(\Sigma_\rho)}=O(\E^{-C(\varepsilon) t}),\quad t\to\infty.\eeq
 
 The matrix $v^{\text{mod}}(k)$ is piecewise constant with respect to $k$. In the next section we study briefly  the respective vector RHP. It was solved in \cite{EGKT}, \cite{EGT}, however the uniqueness was not established there. Moreover, using the trace formula we propose here  a shorter and more transparent way to compute the expansion of $m_1^{\text{mod}}(k)m_2^{\text{mod}}(k)$ as $k\to\infty$, which will approximate the analogous expansion for the initial RHP, because of 
 \beq\label{initrue} m^{(2)}_1(k) m^{(2)}_2(k)= m^{\text{ini}}_1(k) m^{\text{ini}}_2(k), \quad |\im k|>\kappa_N +\rho.\eeq

\section{Unique solution for the vector model RHP}
\label{sec:mp}

\begin{lemma}\label{unique} The following RHP has a unique solution:

\noindent find a vector-valued  function $m^{\text{mod}}(k)=(m_1^{\text{mod}}(k)\ m_2^{\text{mod}}(k))$ holomorphic in the domain $\C\setminus [\I c,  -\I c]$, which is continuous up to the boundary except at points of the set $\mathcal G^{\text{mod}}:=\{ \I c, \I a, -\I a, -\I c\}$, and satisfies the jump condition:
\beq\label{defmvecmod}
m_+^{\text{mod}}(k)= m_-^{\text{mod}}(k) v^{\text{mod}}(k),
\eeq

\beq\label{jumpcondmod}
v^{\text{mod}}(k) = \left\{ \begin{array}{ll}
\begin{pmatrix}
0 & \I \\
\I  & 0
\end{pmatrix},& k\in [\I c, \I a], \,\\
\begin{pmatrix}
0 & -\I \\
-\I & 0
\end{pmatrix},& k\in [-\I a, -\I c],\\
\begin{pmatrix}\E^{-\I \Lambda}& 0\\
0&\E^{\I \Lambda}\end{pmatrix},& k\in [\I a, -\I a],\\
\end{array}\right.
\eeq
the symmetry condition
\beq \label{symcond}
m^{\text{mod}}(-k) = m^{\text{mod}}(k) \sigI,
\eeq
and the normalization condition
\beq\label{normcond}
\lim_{k\to\I\infty} m^{\text{mod}}(k)= (1\ \ 1).\eeq
At any point $\kappa\in\mathcal G^{\text{mod}}$ the vector function $m^{\text{mod}}(k)$ can have at most a fourth root singularity:  $m^{\text{mod}}(k)= O((k-\kappa)^{-1/4}))$, $k\to \kappa$.

\end{lemma}

\begin{proof} We prove the uniqueness. 
Let $m$ and $\hat m$ be two solutions of the RH problem \eqref{defmvecmod}--\eqref{normcond}. Their difference $\tilde m=m-\hat m$ is a holomorphic vector in $\C\setminus [\I c, -\I c]$ which satisfies conditions \eqref{jumpcondmod} and \eqref{symcond} and has the following behavior \[\tilde m(k)= (1\ \ -1)\frac{\ti h}{ k}(1 +o(1))\ \mbox{as}\ \ k\to \I\infty.\] Moreover,
$\tilde m(k)=O((k-\kappa)^{-1/4}))$ as $k\to\kappa$ for $\kappa\in\mathcal G^{\text{mod}}$.

In $\C\setminus [\I c, -\I c]$, introduce a holomorphic function \beq\label{defF}f(k):= \ti m_1(k)\ol {\ti m_1(\ol k)} + \ti m_2(k)\ol {\ti m_2(\ol k)}.\eeq 

Due to \eqref{symcond} this function is  even : $f(-k)=f(k)$
and satisfies
\beq\label{inf} f(k)=\frac{2|\ti h|^2}{k^2}(1+O(k^{-2})),\quad \text {as}\ \ k\to \I\infty;\eeq 
 
\beq\label{near} f(k)=O((k-\kappa)^{-1/2}))\quad \text {as }\ \ k\to\kappa, \ \ \text  {for }\ \ \kappa\in \mathcal G^{\text{mod}}.\eeq   Since $-\ol k=k$ for $k\in\I \R$ and taking into account \eqref{symcond}, for $k\in [\I c, -\I c]$  we get 
\[ \aligned f_+(k)&=\ti m_{1,+}(k)\ol{\ti m_{2,-}(k)}+ \ti m_{2,+}(k)\ol{\ti m_{1,-}(k)},\\
f_-(k)&=\ti m_{1,-}(k)\ol{\ti m_{2,+}(k)}+ \ti m_{2,-}(k)\ol{\ti m_{1,+}(k)},\endaligned\ \quad\ \ \  k\in [\I c, -\I c].\]
By use of \eqref{jumpcondmod}
\[  f_+(k)=\pm \I \left(\left| \ti m_{2,-}(k)\right|^2 +\left| \ti m_{1,-}(k)\right|^2\right)=-  f_-(k)\in \I\R,\quad k\in [\pm \I c, \pm \I a],
\]
\beq  \label{prop}f_+(k)=\E^{-\I\Lambda}  \ti m_{1,-}(k)\ol{\ti m_{2,-}(k)} +\E^{\I\Lambda} \ti m_{2,-}(k)\ol{\ti m_{1,-}(k)}=f_-(k)\in \R,\quad k\in [\I a, -\I a].\eeq
 Thus the function $f(k)$ has no jump on $[\I a, -\I a]$ and is the solution of the following jump problem
 \[ f_+(k)=-f_-(k),\quad k\in [\I c, \I a]\cup [-\I a, -\I c],\] which satisfies \eqref{inf} and \eqref{near}. The unique solution of this problem is given by the formula
\[f(k)=-\frac{2|\ti h|^2}{\sqrt{(k^2 +c^2)(k^2 + a^2)}}.\]
 Therefore, if $\ti h\neq 0$ then $f(0)<0$. But  according to \eqref{defF} and \eqref{prop} we have $f_+(0)=f_-(0)\geq0$. Thus, $\ti h=0$ and hence
\[
\ti m_{1,-}(k)=\ti m_{1,+}(k)=\ti m_{2,+}(k)=\ti m_{2,-}(k)=0,\quad k\in [\I c, \I a]\cup [-\I a, -\I c].
\]
In particular, we see that the jump along $[\I c, \I a]\cup [-\I a, -\I c]$ is removable and
the only solution of this problem is trivial : $\tilde m(k)\equiv 0$.
\end{proof}

Now we recall briefly how to solve problem \eqref{defmvecmod}--\eqref{normcond} (cf. \cite{EGKT}, \cite{EGT}).
Consider the two-sheeted Riemann surface $\mathbb X=\mathbb X(\xi)$ associated with the function \[w(k)=\sqrt{(k^2 + c^2)(k^2 + a^2)},\] defined on $\C \setminus ([-\I c, -\I a] \cup [\I a, \I c])$ with $w(0) > 0$. The sheets of $\mathbb X$ are glued along the cuts  $[\I c,\I a]$ and $[-\I a, -\I c]$. Points on this surface are denoted by $p=(k,\pm )$. To simplify notations we keep the notation  $k=(k,+)$ for the upper sheet of $\mathbb X$.
The canonical homology basis of cycles $\{\bf a, \bf b\}$ is chosen as follows: The $\bf a$-cycle surrounds the points $-\I a,\I a$ starting on the upper sheet from the left side of
the cut $[\I c,\I a]$ and continues on the upper sheet to the left part of $[-\I a, -\I c]$ and returns after changing sheets. The cycle $\bf b$ surrounds the points $\I a, \I c$ counterclockwise on the upper sheet.
Consider the normalized holomorphic differential
\beq\label{omm}
d\om=\Gamma\frac{d\zeta}{w(\zeta)}, \quad \text{where}\ \Gamma:=\left(\int_{\bf a} \frac{d \zeta}{w(\zeta)}\right)^{-1}\in \I \R_-,
\eeq
then $\int_{\bf a}d\om=1$ and 
\beq\label{deftau}\tau=\tau(\xi)=\int_{\bf b} d\om\in \I \R_+.\eeq Let
\[
\theta_3(z\,\big|\,\tau)=\sum_{n\in\Z}\exp\left\{(n^2\tau + 2n z)\pi\I\right\}, \quad z\in\C,
\]
be the  Jacobi theta function. Recall that $\theta_3$ is an even function, $\theta_3(-z\,\big|\,\tau)=\theta_3(z\,\big|\,\tau)$, and satisfies
\[
\theta_3(z+ n + \tau(\xi)\ell\,\big|\,\tau)=\theta_3(z\,\big|\,\tau )\exp\left\{- \pi\I \tau \ell^2 - 2\pi\I\ell z\right\} \mbox{ for } \hspace{3pt} l,n \in \Z..
\]
Furthermore, let $A(p)=\int_{\I c}^p d\om$ be the Abel map on $\mathbb X$. We identify the upper sheet of $\mathbb X$ with the complex plane  with cuts: $\C\setminus([\I c, \I a]\cup [-\I a, -\I c])$,  and put $(k,+)=k$. Allowing only paths of integration in $\C \setminus [\I c, -\I c]$ we observe that $A(k)$ is a holomorphic function in the given domain with the following properties:
\begin{itemize}
\item $A_+(k)=-A_-(k)\quad (\mbox{mod } 1)$ for $k\in [\I c, \I a]\cup [-\I a, -\I c]$;
\item $A_+(k)-A_-(k)=-\tau$ as $k\in [\I a, -\I a]$;
\item $A(-k)=-A(k) + \frac{1}{2}(\mod 1)$ as $k\in\C\setminus [\I c, -\I c]$, 
\item $A_+(\I a)=-\frac{\tau}{2} =-A_-(\I a)$, $A_+(-\I a) = -\frac{\tau}{2} + \frac{1}{2},$  $A_-(-\I a)=\frac{\tau}{2} + \frac{1}{2}.$
\item $A((\infty,+))=\frac{1}{4}$; \ \ $A(k)- A((\infty,+))= -\Gamma k^{-1} + O(k^{-3})$ as $k\to\infty$.
\end{itemize}

On $\C\setminus [\I c, - \I c]$ introduce two functions
\[
\alpha^\Lambda(k)=\theta_3\left(A(k) -\frac{1}{2}- \frac{\tilde\Lambda}{2}\,\big |\,\tau\right)\theta_3\left(A(k) -\frac{\tilde\Lambda}{2}\,\big|\,\tau\right),\]
\[\beta^\Lambda(k)=\theta_3\left(-A(k) -\frac{1}{2}- \frac{\tilde\Lambda}{2}\,\big|\,\tau\right)\theta_3\left(-A(k) -\frac{ \tilde\Lambda}{2}\,\big|\,\tau\right),
\]
where $\tilde\Lambda=\frac{\Lambda}{2\pi}\in\R$ and $A(k)= A((k,+))$ for $k\in\C$. The properties of the Abel integrals listed above imply that the functions $\alpha^0(k)$ and $\beta^0(k)$ have square root singularities at the points $\pm\I a$.
Using the formula (cf. \cite{dubr})
\[
 \theta_3\left(u\,\big|\,\tau\right)\,
 \theta_3\left(u-\frac 1 2\,\big|\,\tau\right)=\theta_3\left(2u-\frac 1 2\,\big|\,2 \tau\right)\,\theta_3\left(\frac 1 2\,\big|\, 2\tau\right),
 \]
 we can represent  functions $\alpha^\Lambda(k)$ and $\beta^\Lambda(k)$ as
 \[
\alpha^\Lambda(k)=\theta_3\left(2A(k) -\frac{1}{2}- \tilde\Lambda\,\big|\,2\tau\right)\,\theta_3\left(\frac 1 2\,\big|\, 2\tau\right),\]
\[\beta^\Lambda(k)=\theta_3\left(-2A(k) +\frac{1}{2}- \tilde\Lambda\,\big|\,2\tau\right)\,\theta_3\left(\frac 1 2\big| \, 2\tau\right).
\]
Introduce the functions
\beq\label{al} \hat\alpha(k):=\frac{\alpha^{\Lambda}(k)}{\alpha^0(k)}=\frac{\theta_3\left(2A(k) -\frac{1}{2}- \tilde\Lambda\,\big|\,2\tau\right)}{\theta_3\left(2A(k) -\frac{1}{2}\,\big|\,2\tau\right)}\eeq
\beq\label{be} \hat \beta(k):=\frac{\beta^{\Lambda}(k)}{\beta^0(k)}=\frac{\theta_3\left(-2A(k) + \frac{1}{2}-\tilde\Lambda\,\big|\,2\tau\right)}{\theta_3\left(-2A(k)+\frac{1}{2}\,\big|\,2\tau\right)}.\eeq
Evidently, both functions $\hat \alpha(k)$ and $\hat \beta(k)$ have square root singularities  at the points $\pm \I a$ if $\tilde\Lambda\notin \Z$.  Moreover,
\[
 \lim_{k\to\infty}\hat \alpha(k)=
 \lim_{k\to\infty }\hat\beta(k)=\frac{\theta_3\left( \tilde\Lambda\,\big|\,2\tau\right)}{\theta_3\left(0\,\big|\,2\tau\right)}.
\]
Due to the first three properties of the Abel map we get
\[
\hat\alpha_+(k)=\hat\beta_-(k)\ \mbox{and}\  \hat \beta_+(k)=\hat \alpha_-(k)\ \mbox{for}\ k\in [\I c,\I a]\cup [-\I a, -\I c],
\]
\[
 \hat \alpha_+(k)
 =\E^{-\I \Lambda}\hat \alpha_-(k)\ \ \ \mbox{and}\ \ \  \hat\beta_+(k)=\E^{\I \Lambda}\hat\beta_-(k)\ \mbox{for}\ k\in[\I a, -\I a],
\]
\[ \hat\alpha(-k)=\hat \beta(k) \ \mbox{for}\ k\in\C\setminus [\I c, -\I c].
\]
Now introduce the function
\beq\label{defgamma}
\ti\gamma(k)=\sqrt[4]{\frac{k^2 + a^2}{k^2+c^2}},
\eeq
defined uniquely on the set $\C\setminus ([\I c, \I a]\cup [-\I a, -\I c])$ by the condition $\arg\ti\gamma(0)=0$. This function satisfies the jump conditions
\[
\begin{array}{ll} \ti\gamma_+(k)=\I\ti\gamma_-(k), & k\in [\I c, \I a],\\
\ti\gamma_+(k)=-\I\ti\gamma_-(k), & k\in [\I a, -\I c].
\end{array}
\]
Thus, the following result is valid:

\begin{lemma} (\cite{EGKT}, \cite{EGT}) Let $\hat\alpha(k)$, $\hat \beta(k)$, $\ti \gamma (k)$ be defined by formulas \eqref{al}-\eqref{defgamma}. Then the vector function
\beq\label{mmod}
m^{\text{mod}}(k)=
\left(\ti\gamma(k)\frac{\hat \alpha(k)}{\hat\alpha(\infty)},\ \ti\gamma(k)\frac{\hat \beta(k)}{\hat\beta(\infty)}\right)
\eeq
solves  problem \eqref{defmvecmod}--\eqref{normcond}.
\end{lemma}

Note that both components of the vector-valued function $m^{\text{mod}}(k)$ are bounded everywhere except for small vicinities of the points of the set $\mathcal G^{\text{mod}}$,
where they have singularities of the type $(k-\kappa)^{-1/4}$, $\kappa\in\mathcal G^{\text{mod}}$.

\begin{remark}\label{rem33} We observe that
\[ \hat\alpha_\pm(0)=\frac{\theta_3\left(\mp\tau -1- \tilde\Lambda\,\big|\,2\tau\right)}{\theta_3\left(\pm\tau +1\,\big|\,2\tau\right)},\quad  \hat \beta_\pm(0):=\frac{\theta_3\left(\pm\tau + 1-\tilde\Lambda\,\big|\,2\tau\right)}{\theta_3\left(\pm\tau+1\,\big|\,2\tau\right)}.
\]
This means that for $\tilde\Lambda=\frac{1}{2} \, \quad(\mbox{mod } n)$ we have $m_\pm^{\text{mod}}(0)=(0,0)$. From Theorem \ref{theor1} it follows then that 
for $\Lambda= 2\pi \ti\Lambda=\pi (2n+1)$, $n\in\Z$ the matrix model RHP associated with the jump \eqref{jumpcondmod} does not have an invertible solution.
\end{remark}

Recall now that we constructed the solution for the jump problem \eqref{jumpcondmod} with $\ti\Lambda=\frac{\Lambda}{2\pi}$ and $\Lambda$ given by formula \eqref{La}.  Due to \eqref{symcond}, the asymptotic expansion of the vector components product should be the following:\beq\label{osnova} m_1^{mod}(k)m_2^{mod}(k)= 1 +\frac{q^{mod}(x,t,\xi)}{2k^2}+O(k^{-4}).\eeq

Let us show that in fact for any fixed $\xi$  coefficient  $q^{mod}(x,t,\xi)$  represents  the classical one-gap solution for the KdV equation associated with the spectrum $\mathfrak G(\xi)$ (cf. \eqref{spectr}) and with the initial Dirichlet divisor $p_0$
defined uniquely by the Jacobi inversion (compare \eqref{defp0}, \eqref{deltaini}):
\beq\label{invJ} \int_{-a^2}^{p_0}d\hat \omega= \I \Delta,\quad p_0=(\la(0,0,\xi), \pm). \eeq
Here $d\hat\omega$ is the normalized holomorphic Abel differential of the first kind on the elliptic Riemann surface $\mathbb M=\mathbb M(\xi)$ associated with the function \[\mathcal R(\lambda,\xi) = \sqrt{\lambda (\lambda +c^2)(\lambda+a^2(\xi))},\] with cuts along the spectrum.

Let ${\bf\hat b}$, $\bf{\hat a}$ be the canonical basis on $\mathbb M$, where the cycle ${\bf\hat b}$  surrounds the interval $[-c^2, -a^2]$ counterclockwise on the upper sheet and the cycle ${\bf\hat a}$ supplements ${\bf\hat b}$ by passing along the gap $[-a^2, 0]$ in the positive direction on the lower sheet and then changing the sheet. The normalization for $d\hat \omega$ is given by formula $\int_{\bf\hat a}d\hat\omega=2\pi\I$. 

Denote $\int_{\bf\hat b}d\hat\omega=\hat\tau$. It is straightforward to check that $\hat\tau=4\pi\I\tau$ (cf. \eqref{deftau}).

Furthermore, let $\hat A(p) := \int_\infty^p d\hat \omega$ be the associated Abel map and \[\mathcal K := -\hat A (-a^2) = -\frac{\hat \tau}{2} - \pi \I\] be the Riemann constant. Introduce the  wave and frequency numbers $V=V(\xi)$ and $W=W(\xi)$ (\cite{KUK}, \cite{MAR}), which are ${\bf\hat b}$ - periods of the normalized Abelian differentials of the second kind $d\Omega_1$ and $d\Omega_3$ on $\mathbb M$
uniquely defined by the order of the pole at infinity
\[
d\Omega_1 = \frac{\I}{2\sqrt\la}(1+ O(\la^{-1}))d\la,
\hspace{10pt}
d\Omega_3 = -\frac{3\I}{2}\sqrt\la (1 +O(\la^{-1}))d\la,\ \  \la\to\infty
,\] 
and by the normalization conditions
$
\int_{\bf \hat a} d \Omega_{1,3} = 0.
$
Thus,
\[\I V := \int_{\bf \hat b} d\Omega_1, \hspace{20pt} 
\I W := \int_{\bf \hat b} d\Omega_3. 
\]
The following result is obtained in
 \cite{EGT}.
\begin{lemma}\label{EGT} Let $B=B(\xi)$ be as in Lemma \ref{lemg}, ${\bf (c)}$ and $\Gamma=\Gamma(\xi)$ be given by \eqref{omm}. Then 
the following identities hold
\[
t B = V x - 4 W t,
\quad
4 \pi \I \Gamma  = -V.
\]
\end{lemma}
Recall now that the one-gap solution corresponding to  the spectrum $\mathcal G(\xi)$ and to the initial divisor \eqref{invJ},  can be expressed by  the trace formula:
\beq\label{qmodl}
q^{\text{per}}(x,t,\xi)=-c^2 - a^2 -2\la(x,t,\xi),
\eeq
where $\la(x,t)=\la(x,t,\xi)\in [-a^2, 0]$ is the projection of  $p(x,t)=(\la(x,t), \pm)\in \mathbb M$, which is  the unique solution of the Jacobi inversion problem
\beq\label{zhut}
\int_{p_0}^{p(x,t)}d\hat\omega=\I (V x - 4W t)(\mbox{mod} \ 2\pi \I).\eeq
Due to \eqref{defp0} we can also represent it as
\[\int_{-a^2}^{p(x,t)}d\hat\omega=\I (V x - 4W t +\Delta).\]
Evidently $\la(x,t)=0$ corresponds to the local minimum of $q^{per}(x,t)$.
Indeed,
\[\la(x,t)=0\ \ \mbox{iff}\ \  
\frac{V x - 4W t  + \Delta}{2\pi}=\frac{1}{2} \quad(\mbox{mod}\ \mathbb Z).\]

Let us now compare function $q^{per}(x,t,\xi)$ with the second term of the expansion for the product $m_1^{mod}(k)m_2^{mod}(k)=:p(k)$, which is given by formula (see \eqref{al}, \eqref{be}):
\[p(k)=\ti\gamma^2(k)\frac{\theta_3 (2A(k) -\frac{1}{2}- \tilde\Lambda)\,\theta_3(-2A(k) + \frac{1}{2}-\tilde\Lambda)\,\theta_3(0)^2}{(\theta_3(2A(k) -\frac{1}{2}))^2(\theta_3( \tilde\Lambda))^2}.\]
To this end we first prove
\begin{lemma}
The function $p(k)$, $k\in \C$, admits the following representation:
\beq\label{green}
p(k)=\frac{k^2 - \la(x,t)}{\sqrt{(k^2  + a^2)(k^2 + c^2)}}.
\eeq
\end{lemma}

\begin{proof}
Given \eqref{al} and \eqref{be}, consider  the function 
\[\ti p(k)=p(k)\ti\gamma^{-2}(k)=\frac{\hat\alpha(k)\hat\beta(k)}{\hat\alpha(\infty)\hat\beta(\infty)}.\]
By the symmetry property  we have $\ti p(-k)=\ti p(k)$.  Moreover,  this function does not have jumps for $k\in [-\I c, \I c]$, and $\ti p(k)\to 1$ as $k\to \infty$. Thus, it must be a meromorphic  (in fact, rational) function of $\la=k^2$ in the whole complex plane. Due to \eqref{invJ} and \eqref{zhut} the function $\hat\alpha(k)\hat\beta(k)$ has the only  zero, simple  with respect to $\la$, at the point $\la=\la(x,t)$, and the only simple pole (again with respect to $\la$) at $\la = -a^2$. We conclude that \[\ti p(k)=
\frac{\hat\alpha(k)\hat\beta(k)}{\hat\alpha(\infty)\hat\beta(\infty)}=\frac{k^2 - \la(x,t)}{k^2 + a^2},\] 
which together with \eqref{defgamma} implies \eqref{green}.
\end{proof}
In turn, decomposing  \eqref{green} with respect to $\frac{1}{2k^2}$ we get  the same   trace formula \eqref{qmodl} for $q^{\text{mod}}(x,t,\xi)$ in \eqref{osnova}. It proves that \beq\label{osnova2}q^{\text {mod}}(x,t,\xi)= q^{\text{per}}(x,t,\xi).\eeq Moreover, the property of the combination 
of theta functions involved in $m^{mod}_1m^{mod}_2$ to be a rational function of the spectral parameter $\la=k^2$ is tightly connected with the analogous property of the product of two branches of the Baker--Akhiezer function. It allows us  to expect that this approach may considerably simplify the evaluation of asymptotics in the case of finite gap backgrounds.

\section{The solution of the model matrix RHP and its properties}

In this section we propose a proper matrix model solution with a nonintegrable singularity at point $k=0$.

\begin{theorem} \label{mainr}
There exists a matrix model solution $M^{mod}(k)$ of the model RHP which satisfies the following properties:
\begin{enumerate} [(1)]
\item It is holomorphic in $\C\setminus [\I c, -\I c]$, continuous up to the sides of the  contour $[\I c, - \I c]$ except of points $\mathcal G^{\text{mod}}\cup\{0\}$;
\item  At points of $\mathcal G^{\text{mod}}$ it has weak singularities,  $M^{mod}(k)=O(k-\kappa)^{-1/4}$ as $k\to\kappa\in\mathcal G^{\text{mod}}$, and
$M^{mod}(k)=O(k^{-1})$ as $k\to 0$\footnote{we can not call it the pole, because the matrix has a jump in this point};
\item It possess the symmetry property: 
\beq\label{symm} M^{mod}(-k)=\sigma_1 M^{mod}(k)\sigma_1;\eeq
\item It satisfies the normalization property
\beq\label{asyM}M(k)\to \id, \quad k\to\infty.\eeq
\item $\det M^{mod}(k) = 1$ for all $k\in \C$;
\item 
The vector $m^{(2)}(k)[M^{mod}(k)]^{-1}$ is a bounded continuous  function in a vicinity $\mathcal O$ of point $k=0$;
\end{enumerate}
\end{theorem}
We preface the proof of this theorem by the following
\begin{lemma} \label{sos}There exists a  vector solution $\nu(k)=(\nu_1(k), \nu_2(k))$ to the jump problem \eqref{jumpcondmod}  which satisfies:  \begin{itemize} 
\item The symmetry condition $\nu_1(k)=\nu_2(-k)$, $k\in\C\setminus [\I c, -\I c]$; \item The  asymptotical behavior :
\beq\label{asy} \nu(k)=\I k (-1,\ \ 1) (1+O(\frac{1}{k})), \quad k\to\infty.\eeq
\item 
Vector $\nu(k)$ is a holomorphic vector function in $\C\setminus [\I c, -\I c]$, continuous up to the boundary except of points $\mathcal G^{\text{mod}}$, where a fourth root singularities are admissible.
\end{itemize}
\end{lemma}
\begin{proof}
From Lemma \ref{EGT} it follows that vector $\nu$ solves the jump problem $\nu_+(k)=\nu_-(k) v^{\text{mod}}(k)$ where
\beq\label{jumpcondmod1}
v^{\text{mod}}(k) = v^{\text{mod}}(k,x,t,\xi)= \left\{ \begin{array}{ll}
\I\sigma_1 ,& k\in [\I c, \I a], \,\\[3 mm]
-\I\sigma_1,& k\in [-\I a, -\I c],\\[3 mm]
\E^{(4\I W(\xi) t -\I V(\xi) x -\I\Delta(\xi))\sigma_3},& k\in [\I a, -\I a].
\end{array}\right.
\eeq
From formulas (2.3) and (2.6) of \cite{EGT} it follows also that: 
\[\aligned \I V(\xi)& =Z_+(k) - Z_-(k), \ \text{for} \  k\in [\I a, - \I a], \\
Z(k) & := Z(k,\xi)=\I\int_{\I c}^k \frac{(s^2 -h)ds}{\sqrt{(s^2 + c^2)(s^2 + a^2)}},\\ h&=\int_{\I a}^0 \frac{s^2 ds}{\sqrt{(s^2 + c^2)(s^2 + a^2)}}\left(\int_{\I a}^0 \frac{ds}{\sqrt{(s^2 + c^2)(s^2 + a^2)}}\right)^{-1}.\endaligned \]
Recall also that \[Z_+(k)+ Z_-(k)= 0\  (\text{mod} \ 2\pi\I), \quad k\in[\I c, \I a]\cup [-\I a, -\I c].\]
In fact \[Z(k)=\I\int_{-c^2}^{k^2}\frac{\la - h}{2\mathcal R(\la)}d\la\]
is the classical quasimomentum associated with the  Riemann surface $\mathbb M(\xi)$.
Thus,
\[ v^{\text{mod}}(k,x,t,\xi)= \E^{\left(4\I W(\xi) t +\left(Z_+(k) - Z_-(k)\right) x -\I\Delta(\xi)\right)\sigma_3}, \quad k\in [\I a, -\I a].\]
We see that the vector
\beq\label{trans}\mathcal S(k):=\mathcal S(k,x,t,\xi)=m^{\text{mod}}(k,x,t,\xi)\E^{-Z(k,\xi) x\,\sigma_3}\eeq
solves the jump problem
$\mathcal S_+(k)=\mathcal S_-(k)v^{\mathcal S}(k),$
\[v^{\mathcal S}(k)=v^{\mathcal S}(k, t, \xi)=\left\{ \begin{array}{ll}
\I\sigma_1 ,& k\in [\I c, \I a], \,\\[2 mm]
-\I\sigma_1,& k\in [-\I a, -\I c],\\[2 mm]
\E^{(4\I W(\xi) t  -\I\Delta(\xi))\sigma_3},& k\in [\I a, -\I a].
\end{array}\right.
\]

Let us treat the variables $x$, $t$, $\xi$ as independent variables.  Then $\frac{\partial}{\partial x} v^{\mathcal S}(k)=0$ and the vector \[\aligned \hat  S(k)& =\frac{\partial}{\partial x}\mathcal S(k,x,t,\xi)\\
&=\left((\frac{\partial}{\partial x} m_1^{\text{mod}}(k) - Z(k)m_1^{\text{mod}}(k))\E^{-Z(k)x},\ \
(\frac{\partial}{\partial x} m_2^{\text{mod}}(k) + Z(k)m_2^{\text{mod}}(k))\E^{Z(k)x}\right) \endaligned\]
solves the same jump problem as $\mathcal S(k)$:
\[ \hat S_+(k)=\hat S_-(k)v^{\mathcal S}(k).\]

Going back with the conjugation inverse to \eqref{trans}, applied to vector $\hat S$, we conclude that the vector

\[\aligned & \nu(k):=\hat S(k)\E^{Z(k)x\,\sigma_3}\\
&=\left(\frac{\partial}{\partial x} m_1^{\text{mod}}(k) - Z(k)m_1^{\text{mod}}(k),\ \
\frac{\partial}{\partial x} m_2^{\text{mod}}(k) + Z(k)m_2^{\text{mod}}(k)\right), \endaligned\]
solves the model RHP \eqref{jumpcondmod1}, which is the same as \eqref{jumpcondmod}.

Next, since $Z(k)=\I k (1 +O(k^{-1})$ as $k\to\infty$, it is easy to see that \eqref{asy} is fulfilled. The singularities of $\nu(k)$ at the  points of $\mathcal G^{\text{mod}}$ are the same as for $m^{\text{mod}}(k)$. This follows from formulas \eqref{al}-\eqref{mmod}  and the fact, that the differentiation $\frac{\pa}{\pa x} m^{\text{mod}}(k)$ does not affect the part of denominators in \eqref{mmod}, which are responsible for singularities. Indeed, for example for the first component
\[\frac{\pa}{\pa x} m^{\text{mod}}_1(k)=\ti\gamma(k)\frac{V(\xi)}{2\pi}\frac{\theta_3(0\,\big |\,2\tau)}{\theta_3(2A(k) -\frac{1}{2}\,\big |\,2\tau)}\ \frac{d}{d \ti \Lambda}\left(\frac{\theta_3\left(2A(k) -\frac{1}{2}-\ti\Lambda\,\big |\,2\tau\right)}{\theta_3\left(\ti\Lambda\,\big |\,2\tau\right)}\right),\]
because $ \frac{\pa \ti\Lambda}{\pa x}=\frac{V(\xi)}{2\pi}.$
\end{proof}

\begin{corollary} Vector function $\ti \nu(k)=\frac{\nu(k)}{\I k}$ solves the model RHP \eqref{jumpcondmod} and  satisfies the antisymmetry condition \beq\label{antisym}\ti \nu_1(-k)=-\ti\nu_2(k),\eeq moreover
\beq\label{normell}\ti\nu(k)\to (-1,\ 1),\quad k\to\infty.\eeq
It is holomorphic outside the contour $[\I c, -\I c]$, has the fourth root singularities at $\mathcal G^{\text{mod}}$ and a singularity $\ti\nu(k)=O(k^{-1})$ as $k\to 0$.
\end{corollary}
{\it Proof of theorem \ref{mainr}}. 
Set 
\[M^{\text{mod}}(k):=\frac{1}{2}\begin{pmatrix} m_1^{\text{mod}}(k) -\ti\nu_1(k)& m_2^{\text{mod}}(k) -\ti\nu_2(k)\\[3 mm]
m_1^{\text{mod}}(k) +\ti\nu_1(k)& m_2^{\text{mod}}(k) +\ti\nu_2(k)\end{pmatrix}.\]
Evidently, it solves the model jump problem.
Equality \eqref{antisym} guaranties the structure
\[M^{\text{mod}}(k)=\frac{1}{2}\begin{pmatrix}\psi_1(k)&\psi_2(k)\\ \psi_2(-k)&\psi_1(-k)\end{pmatrix},\]
and, therefore \eqref{symm}. Equality \eqref{asyM} follows from \eqref{normell}.
Singularities described by item {\it(2)}  are evident.

Let us discuss the invertibility of $M^{\text{mod}}(k)$.
Put $s(k):=\det M^{\text{mod}}(k)$. Computing it, we get  \[s(k)=\frac{m_1(k)\nu_2(k)-\nu_1(k)m_2(k)}{2\I k},\] where $\nu(k)$ is defined in Lemma \ref{sos}. Evidently, $s(k)$ does not have jumps. It is meromorphic with the only  possible pole at $k=0$, and it is bounded at infinity : $\lim_{k\to\infty}s(k)=1$. Thus, 
we get $s(k)=1 +\frac{C}{k}$, where $C$ is a constant. But we also know that it is even:  $s(-k)=s(k)$, i.e.\ in fact $C=0$ and $\det M^{\text{mod}}(k)\equiv 1$. This proves item {\it(5)}.

It remains to prove item {\it(6)}. We have
\[\aligned
\left[M^{\text{mod}}(k)\right]^{-1}&=\frac{1}{2}\begin{pmatrix}\psi_1(-k)&- \psi_2(k)\\ -\psi_2(-k)&\psi_1(k)\end{pmatrix}\\ 
&=\frac{1}{2}\begin{pmatrix} m_1^{\text{mod}}(-k) -\ti\nu_1(-k)& -m_2^{\text{mod}}(k) +\ti\nu_2(k)\\[3 mm]
-m_1^{\text{mod}}(k) -\ti\nu_1(k)& m_2^{\text{mod}}(-k) +\ti\nu_2(-k)\end{pmatrix}.
\endaligned\]
Put now $f(k):=m^{(2)}(k)[M^{mod}(k)]^{-1}=(f_1(k), f_2(k))$, $k\in\mathcal O$, where $\mathcal O$ is a small vicinity of point $k=0$ with $\mathrm{diam}\, \mathcal O<\rho$ (that is $\mathcal O$ is located inside the strip between $\mathcal C$ and $\mathcal C^*$).

Since $f(k)$ does not have jumps in $\mathcal O$, and we have a symmetry $f(-k)=f(k)\si_1$, then it is sufficient to prove for its first component the following

\begin{lemma}
The function $f_1(k)$ has a removable singularity at point $k=0$.
\end{lemma}

\begin{proof}
The singularity at point $k=0$ is a simple pole for $f_1$.  It means that it is sufficient to prove that   $f_1(k)=o(k^{-1})$ from any fixed direction. As an appropriate direction we take the real positive ray $k>0$. We use a trivial fact that if $k\to 0$ then $-k\to 0$. 

To simplify notations, put $\ti m(k)=m^{(2)}(k)$, $m(k)=m^{\text{mod}}(k)$. 
Then
\[\ti m_1(k)\to \ti m_{1,+}(0),\quad  \ti m_1(-k)\to\ti m_{1,-}(0),\quad  \nu_1(k)\to \nu_{1,+}(0), \quad \nu_1(-k)\to \nu_{1,-}(0),\]
and 
\[ \ti m_{1,+}(0)\nu_{1,-}(0)= \ti m_{1,-}(0)\nu_{1,+}(0),\]
because the jump in $\mathcal O$ for model RHP is diagonal, moreover, this jump is the same  in $\mathcal O$ for $m^{(2)}$ and $\nu$.
According to symmetries
\[ \psi_2(-k)=m_1(k) +\frac{\nu_1(k)}{\I k}.\]
Then
\[\aligned f_1(k)& =\frac{1}{2}\left(\ti m_1(k)\psi_1(-k)-\ti m_2(k)\psi_2(-k)\right)\\
&= \frac{1}{2\I k}\left(\nu_1(-k)\ti m_1(k) - \nu_2(-k)\ti m_2(k)\right) + O(1), \quad k\to 0. \endaligned\] 
But  
\[\aligned & \nu_1(-k)\ti m_1(k) - \nu_2(-k)\ti m_2(k)=\nu_1(-k)\ti m_1(k)- \nu_1(k)\ti m_1(-k)\to \\
&  \ti m_{1,+}(0)\nu_{1,-}(0)- \ti m_{1,-}(0)\nu_{1,+}(0)=0,\quad k\to 0,\quad k\in \R_+.\endaligned\]
\end{proof}

\begin{corollary}
The vector $m^{(2)}(k)[M^{mod}(k)]^{-1}$ is holomorphic in $\mathcal O$.
\end{corollary}

This proves theorem \ref{mainr}.

\section{The matrix solution of the parametrix problem}\label{sec7}

In this section we study the matrix solutions of the local RHPs in vicinities of the points $\pm \I a$. Consider first the point $\I a$. Let $\mathcal B$ be a vicinity of this point as it was introduced at the end of Section \ref{sec4}.  Introduce in $\mathcal B$ a local change of variables
 \beq\label{double ve}
 w^{3/2}(k)=-\frac{3\I t}{2}(g(k) - g_\pm(\I a)), \quad k \in \mathcal B,
 \eeq
 with the cut along the interval $J:=[\I c, \I a]\cap\ol{\mathcal B}$. We observe that 
 \beq\label{imp23}
 w^{3/2}(k)=P(a) \E^{\frac{3\pi \I}{4}} t (k-\I a)^{3/2}( 1 +O(k-\I a)), \quad P(a)>0.
 \eeq
 Indeed, from \eqref{deffung} and Lemma \ref{lemg} it follows that for $\I s\to \I a\pm 0$
 \[\aligned
 \re (-\I g(\I s))&=12\int_{a\pm 0}^s \left(\frac{c^2 - a^2}{2} +\xi - s^2\right)\sqrt{\frac{a+s}{c^2 - s^2}}\sqrt{a-s}\,ds\\ 
 &=-8\left(\frac{c^2 - 3 a^2}{2} +\xi\right)\sqrt{\frac{2a}{c^2 - a^2}}(a-s)^{3/2}(1 +O(a-s)).
 \endaligned\]
Since $a(\xi)$ is a monotonous function with $a(\frac{c^2}{3})=c$ and $a(-\frac{c^2}{2})=0$, this implies \eqref{imp23} with $P(a)>0$.
Thus, $w(k)$ is a holomorphic function in $\mathcal B$ with $w(\I a)=0$,\ $w^\prime(k)\neq 0$.   

Till now we did not specify a particular shape of the boundary  $\pa\mathcal B$ and the shape of the contour $\Sigma^{(2)}$ (cf. \eqref{hatsi} inside $\mathcal B$. Treating $w(k)$  as a conformal map, let us think of $\mathcal B$ as a preimage of a disc $\mathcal O$ of radius $P^{2/3}(a)\rho t^{2/3}$ centred at the origin. Since 
$w(k)=P_1(a)t^{2/3}(\I k +a)(1 +o(1))$, the function $w(k)$ maps the interval $[\I a, \I c]\cap\mathcal B$ into the negative half axis. We can always choose the contours
$\Sigma^{2)}\cap\mathcal B$ to be contained in the pre-image of the rays $\arg w=\pm\frac{2\pi \I}{3}$. 

Next, in $\mathcal B$ introduce the function
\beq\label{defrr}
r(k):=\frac{\sqrt{X(k)}}{F(k)}\,\E^{\mp \frac{\I\pi}{4}}\,\E^{\frac{\mp \I t B}{2}},\quad k\in\mathcal B\cap\{k: \pm \re k>0\} ,\eeq
where $X$ and $F$ are defined by \eqref{RX} and \eqref{forF} respectively, and $B=-2g_+(\I a)$.  By \eqref{thesame} and Lemma \ref{lemF} we conclude that
\[ r_+(k)=\frac{\sqrt{|\chi(k)|}}{F_+(k)}\,\E^{-\frac{\I t B}{2}},\quad r_-(k)=\frac{\sqrt{|\chi(k)|}}{F_-(k)}\,\E^{\frac{\I t B}{2}},\quad k\in [\I c, 0]\cap\mathcal B.\]
Therefore, 
\beq\label{imp45}
r_+(k)r_-(k)=1, \quad k\in J;\quad r_+(k)=r_-(k)\E^{-\I \Delta -\I t B},\quad k\in J^\prime,
\eeq
where  $J=[\I c, \I a]\cap\ol{\mathcal B}$ and
\[ J^\prime:=[\I a, \I b]=[\I a, 0]\cap \ol{\mathcal B}.\]
Denote also \[\mathcal L_1=\Sigma_1\cap \ol{\mathcal B}\cap\{  \re k\geq 0\};\quad 
\mathcal L_2=\Sigma_1\cap \ol{\mathcal B}\cap\{\  \re k\leq 0\}.\]

\begin{figure}[h]
\begin{tikzpicture}
\draw (0,0) {};
\draw (2.5,-2) ellipse (2cm and 1.5cm);
\draw (5,-1.5) arc (120:60:20mm);
\draw (9.5,-2) circle (1.8cm);

\draw (7.7,-2) -- (11.3,-2);
\draw (9.5,-2) -- (8.5, -0.5);
\draw (9.5,-2) -- (8.5, -3.5);

\draw (2.5,-0.5) -- (2.5, -3.5);
\draw (2.5,-2) to[out=35,in=230] (3.8,-0.85);
\draw (2.5,-2) to[out=145,in=-50] (1.2,-0.85);

\draw (4.5,-3) node {$\partial \mathcal B$};
\draw (6, -1) node {$w$};
\draw (2.7,-1.2) node {$J$};
\draw (2.75,-2.7) node {$J'$};
\draw (3.5, -1.6) node {$\mathcal L_1$};
\draw (1.5, -1.6) node {$\mathcal L_2$};
\draw (11.5,-3) node {$\partial \mathcal O$};
\draw (2.8, -2.1) node {$\I a$};
\draw (9.6, -2.3) node {$0$};

\draw [->, thick] (2.5,-1.4) -- (2.5, -1.5);
\draw [->, thick] (2.5,-2.5) -- (2.5, -2.6);
\draw [->, thick] (3.38, -1.3) -- (3.48, -1.2);
\draw [->, thick] (1.52, -1.2) -- (1.62, -1.3);

\draw [->, thick] (6.8, -1.4) -- (7, -1.5);

\draw [->, thick] (8.5, -2) -- (8.6, -2);
\draw [->, thick] (10.3, -2) -- (10.4, -2);
\draw [->, thick] (9.03, -1.3) -- (8.9, -1.1);
\draw [->, thick] (8.9, -2.9) -- (9.03, -2.7);

\draw [->, thick] (4.5, -2) -- (4.5,-1.9);
\draw [->, thick] (11, -1) -- (10.93,-0.9);
\end{tikzpicture}
\caption{The local change of variables $w(k)$.}
\end{figure}

Recall that the vector function $m^{(2)}(k)$ satisfies the jump condition $m_+^{(2)}(k)=m^{(2)}_-(k) v^{(2)}(k)$, with the jump matrix \eqref{jumpcond25}.
Now redefine $m^{(2)}(k)$ inside the domains $\mathcal B$ and $\mathcal B^*$ by
 \beq\label{lastdef}
 m^{(3)}(k)=\left\{ \begin{array}{ll}m^{(2)}(k)[r(k)]^{-\sigma_3}, & k\in \mathcal B,\\
 m^{(3)}(-k)\sigma_1, & k\in \mathcal B^*,\\
 m^{(2)}(k), & k\in \C\setminus(\ol{\mathcal B}\cup \ol{\mathcal B^*}).\end{array}\right.
 \eeq
 By use of \eqref{imp45} we get  $m^{(3)}_+(k)=m^{(3)}_-(k)v^{(3)}(k)$ with
 \beq\label{v4}
 v^{(3)}(k)=\left\{\begin{array}{ll}\begin{pmatrix} 1&0\\ \I\E^{-4/3 w(k)^{3/2}} &1\end{pmatrix}, & k\in J^\prime,\\
 \I\sigma_1, & k\in J,\\
 \begin{pmatrix} 1& \I\E^{4/3 w(k)^{3/2}}\\ 0 & 1\end{pmatrix}, & k\in \mathcal L_1,\\
  \begin{pmatrix} 1& -\I\E^{4/3 w(k)^{3/2}}\\ 0 & 1\end{pmatrix}, & k\in \mathcal L_2,\\
r(k)^{-\sigma_3}, & k\in \partial \mathcal B,\\
 \sigma_1 [v^{(3)}(-k)]\sigma_1,& k\in \partial \mathcal B^*\cup \Sigma^*_{\mathcal B},\\
v^{(2)}(k),& k\in \Sigma^{(2)}\setminus(\Sigma^*_{\mathcal B}\cup \Sigma_{\mathcal B}),\end{array}\right.
\eeq 
where $\Sigma^{(2)}$ is defined by \eqref{hatsi},
\beq\label{recont} 
\Sigma_{\mathcal B}=J\cup J^\prime\cup \mathcal L_1\cup\mathcal L_2,\quad \Sigma^*_{\mathcal B}=\{k: -k\in \Sigma_{\mathcal B}\},
\eeq
with the orientation preserving symmetries for starred contours. In particular, $\pa\mathcal B^*$ should be oriented counterclockwise. 

We observe  that    transformation \eqref{lastdef} applied in $\mathcal B$ to the matrix model problem solution,
\beq\label{123}
M(k):=M^{(\text{mod})}(k)[r(k)]^{-\sigma_3},\quad k\in \mathcal B,
\eeq
leads to wiping out of the jump along $J^\prime$, i.e. in $\mathcal B$ the matrix $M$ satisfies the jump condition $M_+(k)=\I M_-(k)\sigma_1$, $k\in J$. Next by \eqref{double ve},
 the function $w^{1/4}(k)$ has the following jump along the interval $J$:
\[w_+^{1/4}(k)=w_-^{1/4}(k) \I,\quad k\in J.\]
Recall that $ \mathcal O=w(\mathcal B)$. It is now straightforward to check that  the matrix 
\[N(w)=\frac{1}{\sqrt{2}}\begin{pmatrix} w^{1/4}& w^{1/4}\\  - w^{-1/4} & w^{-1/4}\end{pmatrix}, \quad w\in \ol{\mathcal O},\]
 solves the jump problem 
 \[N_+(w(k))=\I N_-(w(k))\sigma_1,\quad k\in J.\]
Therefore, in  $ \mathcal B$ we have $M(k)=  H(k) N(w(k))$, where $ H(k)$ is a holomorphic matrix function in $ \mathcal B$. Moreover, since $\det N(w)=\det [r(k)^{\sigma_3}]=1$, we have
 \beq\label{obr3}\det  H(k)=\det M^{\text{mod}}(k)=\det M(k).\eeq
According to \eqref{1} we get then
\beq\label{main55}
M^{mod}(k)= H(k)N(w(k))r(k)^{\sigma_3},\quad k\in \pa\mathcal B.\eeq
Next, by property $\bf{(b)}$ of Lemma \ref{lemg}  $w_+(k)^{3/2} =-w_-(k)^{3/2}$, $k \in J$, that is
\[v^{(3)}(k)=d_-(k)^{\sigma_3} \mathcal S\, d_+(k)^{-\sigma_3},\quad k\in \mathcal B,\] where
\[ \quad d(k):= \tilde d(w(k)), \quad \tilde d(w)=\E^{ 2/3 w^{3/2}},\]
and
\[ \mathcal S=\left\{\begin{array}{ll} \I\sigma_1, & k\in J,\\
\begin{pmatrix}1&0\\ \I& 1\end{pmatrix}, & k\in J^\prime,\\
\begin{pmatrix} 1& \I\\ 0&1\end{pmatrix}, & k\in\mathcal L_1,\\
\begin{pmatrix} 1& - \I\\ 0&1\end{pmatrix}, & k\in\mathcal L_2.\end{array}\right.
\]
Let us consider the constant matrix $\mathcal S$ as the jump matrix on the contour $\Gamma:=w(\Sigma_{\mathcal B})$ (see \eqref{recont}). Let  $\mathcal A(w)$ be the matrix solution of the jump problem 
\[\mathcal A_+(w)=\mathcal A_-(w) \mathcal S,\quad w\in \Gamma,\]
satisfying the boundary condition
\[
\mathcal A(w)=  N(w)\Psi(w)\tilde d(w)^{\sigma_3},\quad w\in\partial \mathcal O,\quad t\to\infty,
\]
where \[\Psi(w)=\id +\frac{C}{w^{3/2}}(1 + O(w^{-3/2})),\quad w\to \infty,\]
is an invertible matrix and
 $C$ is a constant matrix with respect to $w$, $t$ and $\xi$.
The solution $\mathcal A(w)$ can be expressed via the Airy functions and their derivatives in a standard way (see, for example,\cite{dkmvz}, \cite{Bleher} Chapter 3, \cite{GGM} or \cite{AELT}). In particular,  in the domain between the contours $w(J^\prime)$ and $w(\mathcal L_1)$ we have
\[ \mathcal A(w)=\sqrt{2\pi}\begin{pmatrix} -y_1^\prime(w)&\I y_2^\prime(w)\\-y_1(w)& \I y_2(w)\end{pmatrix},\]
where
$y_1(w)=\mathrm{Ai}(w)$ and $y_2(w)=\E^{-\frac{2\pi \I}{3}} \mathrm{Ai}(\E^{-\frac{2\pi \I}{3}} w)$. The precise formula for $\mathcal A(w)$ in the other domains can be obtained by simple multiplication on the jump matrix $\mathcal S$, but it is not important for us. 

Define the matrix  
\[M^{par}(k):=  H(k)\mathcal A(w(k))d(k)^{-\sigma_3},\quad k\in\mathcal B\setminus\Sigma_{\mathcal B}.\] This matrix then solves in $\mathcal B$ the jump problem
\beq\label{jumppar}
M_+^{par}(k)= M_-^{par}(k)v^{(3)}(k),\quad k\in \Sigma_{\mathcal B}=J\cup J^\prime\cup \mathcal L_1\cup\mathcal L_2,
\eeq
and satisfies for sufficiently large $t$ the boundary condition 
\beq\label{main66}
M^{par}_+(k)= H(k) N(w(k))\Psi(w(k))=M(k)\Psi(w(k)),\quad k\in\pa\mathcal B,
\eeq
where $M(k)$ is defined via \eqref{123}, \eqref{defrr}.
In $\mathcal B^*$ we define $M^{par}(k)$ by symmetry 
\[M^{par}(k)=\sigma_1 M^{par}(-k)\sigma_1.\]

\section{Completion of asymptotical analysis}

The aim of this section is to establish that the solution $m^{(3)}(k)$ given by \eqref{lastdef} is well approximated by $\rI M^{\mathrm{par}}(k)$ inside the domain $\mathcal B\cup \mathcal B^*$,
and by $\rI M^{\mathrm{mod}}(k)$ in $\C\setminus(\mathcal B\cup\mathcal B^*)$.
We follow the well-known approach via singular integral equations (see e.g., \cite{dz}, \cite{GT}, \cite{its} Chapter 4, \cite{len}).
Set
\beq\label{dehatm}
\hat m(k)=m^{(3)}(k) (M^{\text{as}}(k))^{-1},\quad M^{\text{as}}(k):=\begin{cases} M^{\mathrm{par}}(k), & k\in(\mathcal B\cup \mathcal B^*),\\
 M^{\mathrm{mod}}(k), & k\in\C\setminus(\mathcal B\cup \mathcal B^*).
\end{cases}
\eeq
Formula \eqref{jumppar} implies that $\hat m$ does not have jumps inside $\mathcal B\cup \mathcal B^*$. Moreover, from \eqref{lastdef}  and item {\it (6)} of Theorem \ref{mainr} this vector is a holomorphic bounded funtion inside the strip between $\mathcal C$ and $\mathcal C^*$. Let us compute the jump of this vector on $\pa\mathcal B$ by use of \eqref{v4},  \eqref{123}, \eqref{main55}, and \eqref{main66}:
\[\aligned\hat m_+&=m_+^{(3)}\left(M^{\mathrm{par}}_+\right)^{-1}=m_-^{(3)}r^{-\sigma_3}\Psi^{-1}M^{-1}_+= m_-^{(3)}\left(M^{\mathrm{mod}}_-\right)^{-1} M^{\mathrm{mod}}_- r^{-\sigma_3}\Psi^{-1}M_+^{-1}\\
& = \hat m_- M^{\mathrm{mod}}_- r^{-\sigma_3}\Psi^{-1} r^{\sigma_3} \left(M^{\mathrm{mod}}_+\right)^{-1} =\hat m_- M_+\Psi^{-1} M_+^{-1} .\endaligned\] Here we took into account \eqref{123} and the fact that  $M^{\mathrm{mod}}$ does not have a jump on $\pa\mathcal B$. Note also that both matrices $M_+(k)$ and $M^{\mathrm{mod}}(k)$ are  bounded with respect to $t$ uniformly on $\pa\mathcal B$. 

Next, the structure of the matrix $\Psi(w(k))$ implies that
\[ \Psi^{-1}(w(k))=\id +\frac{\mathcal F(k,t)}{t (g(k) - g_\pm(\I a))} ,\quad  \|\mathcal F(k,t)\|\leq O(1), \quad t\to\infty,\]
where the matrix norm estimate $O(1)$ is uniform with respect to $k$ on the compact $\pa\mathcal B\cup\pa\mathcal B^*,$ and uniform with respect to $\xi\in \mathcal I_\varepsilon$. 
Hence $\hat m(k)$ solves the jump problem
 \[
\hat m_+(k)=\hat m_-(k)\hat v(k),
\]
where (cf. \eqref{hatsi}, \eqref{estmoro}):
\[
\hat v(k)=\begin{cases}
\id +M(k) \frac{\mathcal F(k,t)}{t (g(k) - g_\pm(\I a))}M(k)^{-1}, & k\in \partial \mathcal B,\\
\sigma_1\hat v(-k)\sigma_1 , & k\in \partial \mathcal B^*,\\
M_-^{\mathrm{mod}}(k)v^{(3)}(k)(M_+^{\mathrm{mod}}(k))^{-1}, & k\in \Sigma_\rho,\end{cases}
\]
and satisfies the symmetry and normalization conditions:
\[
\hat m(k)=\hat m(-k)\sigma_1,\qquad  \hat m\to (1, \ 1), \  \ \ k\to\infty.
\]

Abbreviate  $W(k)=\hat v(k) -\id$. Recall the estimate \eqref{estmoem}. Hence 
\begin{equation}\label{Wi}
W(k)= \begin{cases} \frac{1}{t (g(k) - g_\pm(\I a))}
M_+(k)\,\mathcal F(k,t)\, M^{-1}_+(k), & k\in \partial\mathcal B, \\
\sigma_1 W(-k) \sigma_1,& k\in \partial\mathcal B^*, \\
M_-^{\mathrm{mod}}(k)(v^{(3)}(k)- v^{\text{mod}}(k))(M_+^{\mathrm{mod}}(k))^{-1}, & k\in \Sigma_\rho\setminus [\I\rho, -\I\rho],\\ 0 & k\in [\I\rho, -\I \rho],\end{cases}
\end{equation}
where we treat $v^{\text{mod}}$ as in \eqref{vm}. Thus the error vector $\hat m(k)$ has jumps on the contour 
\[\hat\Sigma=\Sigma_\rho\cup\pa\mathcal B\cup\pa \mathcal B^*\setminus [\I \rho, -\I \rho]\]
only. This contour does not have  lines inside vicinities of singular points $\I a, -\I a, 0$.
We observe that for all $(x,t)\in \mathcal D_\varepsilon$ the matrix $W(k)$ is continuous on any smooth part of the contour $\hat\Sigma$ and bounded with respect to $k$. Moreover, due to \eqref{estmoem} and \eqref{v4} we have \[\|k^j(v^{(3)}(k)-v^{\text{mod}}(k))\|_{L^{p}(\Sigma_\rho\setminus [\I\rho, -\I\rho])} = O(e^{-C(\varepsilon) t}), \quad p \in [1,\infty], \quad j = 0,1,2,\] (the estimates on the higher moments will be used later). Here we took into account that the reflection coefficient $R(k)$ decays as $O(k^{-6})$ under condition \eqref{decay}. Thus, using \eqref{Wi} and \eqref{obr3} we get \begin{lemma}\label{lemimp}
The following estimates hold uniformly with respect to $\xi\in\mathcal I_\varepsilon$ and $(x,t)\in \mathcal D_\varepsilon$:
\begin{equation}
\label{2w}\|k^j W(k) \|_{L^p(\hat\Sigma)} \leq C(\varepsilon) t^{-1 }, \quad p \in [1,\infty], \quad j = 0,1,2.\end{equation}
\end{lemma}

Now we are ready to apply the technique of singular integral equations. Since this is well known (see, for example, \cite{dz}, \cite{GT}, \cite{len})
we will be brief and only list the necessary notions and estimates. 

Let $\mathfrak C$ denote the Cauchy operator associated with $\hat\Sigma$:
\[
(\mathfrak C h)(k)=\frac{1}{2\pi\I}\int_{\hat\Sigma}h(s)\frac{ds}{s-k}, \qquad k\in\C\setminus\hat\Sigma,
\]
where $h= \begin{pmatrix} h_1 & h_2 \end{pmatrix}\in L^2(\hat\Sigma)$. 
Let  $\mathfrak C_+ f$ and $\mathfrak C_- f$ be its non-tangential limiting values from the left and right sides of $\hat\Sigma$ respectively.

As usual, we introduce the operator $\mathfrak C_{W}:L^2(\hat\Sigma)\cup L^\infty(\hat\Sigma)\to
L^2(\hat\Sigma)$ by formula $\mathfrak C_{W} f=\mathfrak C_-(f W)$, where $W$ is our error matrix \eqref{Wi}. 
Then,
\[
\|\mathfrak C_{W}\|_{L^2(\hat\Sigma)\to L^2(\hat\Sigma)}\leq C\|W\|_{L^\infty(\hat\Sigma)}\leq O(t^{-1}),
\] 
as well as
\beq\label{6w}
\|(\id - \mathfrak C_{W})^{-1}\|_{L^2(\hat\Sigma)\to L^2(\hat\Sigma)}\leq \frac{1}{1-O(t^{-1})}
\eeq
for sufficiently large $t$. Consequently, for $t\gg 1$, we may define a vector function
\[
\mu(k) =(1, \ 1) + (\id - \mathfrak C_{W})^{-1}\mathfrak C_{W}\big((1, \ 1)\big)(k).
\]
Then by \eqref{2w} and \eqref{6w}
\begin{align}\nn
\|\mu(k) - (1, \ 1)\|_{L^2(\ti\Sigma)} &\leq \|(\id - \mathfrak C_{W})^{-1}\|_{L^2(\ti\Sigma)\to L^2(\ti\Sigma)} \|\mathfrak C_{-}\|_{L^2(\ti\Sigma)\to L^2(\ti\Sigma)} \|W\|_{L^\infty(\ti\Sigma)}\\
&= O(t^{-1}).\label{estmu}
\end{align}
With the help of $\mu$, \eqref{dehatm} can be represented as 
\[
\hat m(k)=(1, \ 1) +\frac{1}{2\pi\I}\int_{\hat\Sigma}\frac{\mu(s) W(s)ds}{s-k},
\]
and in virtue of \eqref{estmu} and Lemma \ref{lemimp} we obtain as $k\to \I\infty:$
\[
\hat{m}(k) = (1, \ 1) + \frac{1}{2\pi\I } \int_{\hat\Sigma} \frac{(1, \ 1) W(s)}{s-k} ds + E(k),
\]
where 
\[
|E(k)|\leq \frac{C}{\im k}\|W\|_{L^2(\hat\Sigma)}\|\mu (k)- (1, \ 1)\|_{L^2(\hat\Sigma)}\leq \frac{O(t^{-2})}{\im k},
\] 
 where $O(t^{-2})$ is uniformly bounded with respect to $\xi\in \mathcal I_\varepsilon$, $(x,t)\in \mathcal D_\varepsilon$ and $k \rightarrow \I\infty$.
In the same regime $\re k=0,\  \im k \to +\infty$, we have
\begin{align*}
\frac{1}{2\pi\I } \int_{\hat\Sigma} \frac{(1, \ 1) W(s)}{k-s} ds &=  \frac{f_0(\xi,t)}{2\I k t }(1, \ -1) + \frac{f_1(\xi,t)}{2 k^2 t} (1, \ 1)\\
& + O(t^{-1})O(k^{-3}) +O(t^{-2})O(k^{-1}),
\end{align*}
where $f_{0,1}(\xi,t)$ are uniformly bounded for $t \rightarrow \infty$ and $\xi \in \mathcal{I}_\varepsilon$. Furthermore  $O(k^{-s})$ are vector functions depending on $k$ only and $O(t^{-s})$ are as above. Hence,
\[
m^{(3)}(k) = \hat m(k) M^{\text{mod}}(k) 
= m^{\text{mod}}(k) + \frac{f_0(\xi,t)}{2\I k t} (1, \ -1)M^{\text{mod}}(k)
\]
\[
 + \frac{f_1(\xi,t)}{2 k^2 t}m^{\text{mod}}(k) + O(t^{-1})O(k^{-3}) +O(t^{-2})O(k^{-1}).
\]

Now we are in a position to apply \eqref{main1}, making use of \eqref{initrue}, \eqref{osnova},\eqref{osnova2}, \eqref{qmodl}. Note that since all conjugation steps in the vicinity of $\infty$ involved diagonal matrices with determinant $1$, we have for the solution to IVM RHP from Theorem \ref{thm:vecrhp}:
\[
m_1(k)m_2(k) = m_1^{(3)}(k)m_2^{(3)}(k) = m_1^{mod}(k)m_2^{mod}(k) + O(t^{-1})O(k^{-2}).
\]
Here we used that the entries of $M^{mod}(k)$ are uniformly bounded for $\xi \in \mathcal I_\varepsilon$ and that the $k^{-1}$ term disappears by symmetry  \eqref{eq:symcond}.
Theorem \ref{theormain} is proved.

\vskip 5 mm
\noindent{\bf Acknowledgments.} We are grateful to Alexander Minakov for useful discussions. I.E. is indebted to the Department of Mathematics at the University of Vienna for its hospitality and support during the winter of 2022, where this work was done. I.E. is partially supported by the program "Support of priority research and scientific and technical developments" by  the National Academy of Sciences of Ukraine.

\end{document}